\documentclass[journal]{IEEEtran}
\usepackage{amsmath,amsfonts}
\usepackage[linesnumbered,ruled,vlined]{algorithm2e}
\RestyleAlgo{ruled}
\usepackage{array}
\usepackage[caption=false,font=normalsize,labelfont=sf,textfont=sf]{subfig}
\usepackage{textcomp}
\usepackage{stfloats}
\usepackage{url}
\usepackage{verbatim}
\usepackage{graphicx}
\usepackage{cite}
\usepackage{float}
\usepackage{amsmath}
\usepackage{amssymb}
\usepackage{amsthm}

\usepackage{xcolor}
\newcommand{\CB}{\color{black}}

\usepackage[normalem]{ulem}

\newtheorem{definition}{Definition}
\newtheorem{proposition}{Proposition}
\newtheorem{lemma}{Lemma}
\newtheorem{remark}{Remark}

\usepackage{booktabs}

\usepackage{graphicx}
\usepackage[caption=false]{subfig}

\begin{document}

\title{A general partitioning strategy \\ for non-centralized control}

\author{Alessandro~Riccardi, Luca~Laurenti~\IEEEmembership{Member,~IEEE}, and Bart~De~Schutter~\IEEEmembership{Fellow,~IEEE}
\thanks{This project has received funding from the European Research Council (ERC) under the European Union's Horizon 2020 research and innovation program (Grant agreement No. 101018826) -- Project CLariNet. \textit{(Corresponding author: Alessandro Riccardi)}}
\thanks{The authors are with the Delft Center for Systems and Control (DCSC),  Delft University of Technology, 2628 CD Delft, The Netherlands (email: a.riccardi@tudelft.nl; l.laurenti@tudelft.nl; b.deschutter@tudelft.nl).}}


\IEEEpubid{This work has been submitted to the IEEE for possible publication. Copyright may be transferred without notice, after which this version may no longer be accessible.}
\IEEEpubidadjcol


\maketitle

{\begin{abstract} 
Partitioning is a fundamental challenge for non-centralized control of large-scale systems, such as hierarchical, decentralized, distributed, and coalitional strategies. The problem consists of finding a decomposition of a network of dynamical systems into system units for which local controllers can be designed. Unfortunately, despite its critical role, a generalized approach to partitioning applicable to every system is still missing from the literature. This paper introduces a novel partitioning framework that integrates an algorithmic selection of fundamental system units (FSUs), considered indivisible entities, with an aggregative procedure, either algorithmic or optimization-based, to select composite system units (CSUs) made of several FSUs. A key contribution is the introduction of a global network metric, the partition index, which quantitatively balances intra- and inter-CSU interactions, with a granularity parameter accounting for the size of CSUs, allowing for their selection at different levels of aggregation. The proposed method is validated through case studies in distributed model predictive control (DMPC) for linear and hybrid systems, showing significant reductions in computation time and cost while maintaining or improving control performance w.r.t.\ conventional strategies. 

\end{abstract}}


\section{Introduction}  \label{sec:intro}
{\IEEEPARstart{T}{he} technological and scientific progress of the last century allowed the realization of large-scale infrastructures in several domains. These networks of systems are commonly found in the fields of transportation \cite{an_SurveyIntelligentTransportation_2011, chen_ReviewApplicationsAgent_2010}, power generation and distribution \cite{riccardi_BenchmarkApplicationDistributed_2025,riccardi_CodeUnderlyingPublication_2025a, carrasco_PowerelectronicSystemsGrid_2006, molzahn_SurveyDistributedOptimization_2017,pandey_LiteratureSurveyLoad_2013}, and telecommunications \cite{zhang_NetworkedControlSystems_2020}, among others \cite{ocampo-martinez_HierarchicalDecentralisedModel_2012b, cao_OverviewRecentProgress_2013, cortes_CoordinatedControlMultirobot_2017}. These are considered large-scale systems (LSSs) due to their geographical distribution, heterogeneity of components, and complexity of design and operation \cite{kordestani_RecentSurveyLargescale_2021}. With such characteristics, large-scale networks represent a stimulating challenge for control. 
	
Significant improvements in performance, safety, and user experience can be achieved by applying control methodologies. Of particular interest are non-centralized control techniques \cite{siljak_DecentralizedControlComplex_1991, scattolini_ArchitecturesDistributedHierarchical_2009}, solutions developed for the application of real-time control to large-scale networks. The more traditional centralized approach to control is often unsuitable in this case. Limitations arise from the computational complexity introduced by the size of the problem, delays in data communication, operational constraints, and privacy limitations. With non-centralized control, such issues can be mitigated or translated into a more tractable class of problems \cite{maestre_DistributedModelPredictive_2014}. 
	
A fundamental task at the core of every non-centralized control strategy is the definition of the sub-networks of the system that will be considered as individual operational units, i.e.\ suitable for the application of a local controller part of a non-centralized architecture. The act of defining such sub-networks is defined \textit{partitioning} \cite{riccardi_GeneralizedPartitioningStrategy_2024}. From this abstract formulation, it immediately follows that partitioning a network is a complex task that often has to be tailored for the specific application considered. As a matter of fact, a generalized approach to partitioning is currently missing in the literature. On one side, this issue is related to the fact that often when a novel control strategy is being developed, a network partition is assumed to be given a priori; on the other, partitioning might depend on requirements non exclusively related to control objectives. The present article addresses some of these problems, mainly in the former class, laying the foundation for developing generalized partitioning strategies for control, and showing how it is a central problem to address when deploying a non-centralized control strategy. 
	
	\IEEEpubidadjcol
	
In the literature, the partitioning problem is usually approached in one of the two following ways: i) using the \textit{top-down} aggregation, where we decompose a monolithic system into components, which define the control units; ii) through the \textit{bottom-up} approach, where predefined subsystems are merged into larger control units. However, an approach that works for all categories of systems following a systematic procedure, i.e.\ a generally valid technique, is not present yet in literature. Additionally, the problem of defining the number of control units in the partition is often solved by introducing assumptions or heuristics, whereas the role of the size of these units and their hierarchy is not usually a point of investigation. 
	
In this work, we propose a general partitioning strategy for control, a novel approach expanding and improving on our preliminary results \cite{riccardi_GeneralizedPartitioningStrategy_2024, riccardi_CodeUnderlyingPublication_2024} with stronger mathematical foundations, refined algorithms, generalized metrics, and new strategies to approach the problem. The main contributions we propose are: i) a formal definition of graph equivalence, which allows the definition of the concepts of fundamental system units (FSUs), composite system units (CSUu), and control partition; ii) a novel algorithm for the selection of the FSUs for any graph equivalent to a dynamical system; iii) a global metric, the \textit{partition index}, accounting for the interactions in the network, and the size of CSUs through a \textit{granularity} parameter $\alpha$; iv) optimization-based and algorithmic partitioning approaches, for which the partition index can be specifically characterized, and returning the control partition in terms of groups of FSUs depending on the desired level of granularity.
	
The paper is organized as follows. In the remainder of this section, we present the current partitioning approaches in literature, and we briefly recapitulate some graph theory concepts and notation. In Sec.\ \ref{sec:associated_graph} we propose the concept of graph equivalent to a dynamical system. Then, the novel core concepts of FSU, CSU, control partition, aggregation operation, and the proposed structure of the generalized partitioning strategy are presented in Sec.\ \ref{sec:algorithm_atomic_agents}. The algorithm we propose for the selection of the FSUs is detailed in Sec.\ \ref{sec:algorithm-selection}; while the generalized partitioning strategy, including the new concept of partition index, and the optimization-based and algorithmic approaches are presented in Sec.\ \ref{sec:partitioning-strategy}. We conclude with some application examples for partitioning in Sec.\ \ref{sec:example-network-partitioning}, and a case study illustrating the role of partitioning for the distributed model predictive control of networks of linear and hybrid systems in Sec.\ \ref{sec:example-case-study}. All the experiments and software developed for this article are available at the open-source long-term access repository \cite{riccardi_CodeUnderlyingPublication_2025}.

}
\subsection{Literature Survey} \label{sec:literature}
	
Despite the absence of a generalized approach, partitioning is a relevant topic for the control \cite{chanfreut_SurveyClusteringMethods_2021}. Existing partitioning techniques can be classified into the two fundamental approaches presented below.
\subsubsection{Top-Down Approaches}
These methods address the partitioning problem by decomposing a system into smaller, manageable sub-networks based on network structure, interaction strength, or physical properties. These approaches primarily leverage graph-theoretic metrics \cite{schaeffer_GraphClustering_2007, hernandez_ClassificationGraphMetrics_2011}. In this context, the most widely accepted partitioning metric is modularity \cite{fortunato_CommunityDetectionGraphs_2010}, allowing the definition of partitioning techniques referred to as community detection strategies \cite{malliaros_ClusteringCommunityDetection_2013,  segovia_DistributedModelPredictive_2021a, tang_OptimalDecompositionDistributed_2018}. The modularity metric can be used to capture both the information in the state transition matrix of a linear dynamical system, and the topological information about the network. The scope of modularity-based techniques is to maximize the modularity. Maximization of modularity is known to be an NP-hard problem \cite{blondel_FastUnfoldingCommunities_2008a, schaeffer_GraphClustering_2007}; therefore, polynomial-time algorithmic procedures have been developed for partitioning according to modularity maximization. 

Top-down partitioning approaches have also been developed in specialized contexts for selected applications. For example, in \cite{siniscalchi-minna_NoncentralizedPredictiveControl_2020b} a wind farm is partitioned into control units according to the coupling given by the wake-effect that influences down-stream turbines. In \cite{barreiro-gomez_PartitioningLargescaleSystems_2019}, flow-based distribution systems are considered, focusing on water distribution networks. One relevant feature of the approach proposed in that work is multi-criteria optimization for partitioning. 

\subsubsection{Bottom-Up Approaches}

These methods are aggregative and start from the assumption that a group of autonomous and indivisible FSUs is given a priori. This assumption will generally hold for cooperative groups of mobile agents, such as coalitions of mobile robots, fleets of aerial vehicles, and platoons of autonomous driving cars. However, bottom-up approaches are also used in transportation and power networks.

One recent distributed control setting is coalitional control \cite{fele_CoalitionalControlCooperative_2017b,fele_CoalitionalControlSelfOrganizing_2018} where control theory and game theory are combined. In coalitional control, FSUs aggregate into bigger groups, denoted coalitions, to pursue a common control objective. For each coalition, i.e.\ collection of FSUs, a performance index is computed, which is a function of the aggregated state and input values of the coalition. This index can be interpreted as a transferable utility that can be reallocated among the agents. Partitioning is achieved through an algorithmic procedure, where agents are exchanged among coalitions to maximize the differential gain in the coalitional cost until no further improvement is possible. 

Bottom-up approaches are generally more suited for situations where the definition of the sub-networks is non-static, which can occur in the case of non-stationary networks, where a partitioning technique is required to produce a non-stationary definition of the groups. An example of this approach is in \cite{ananduta_OnlinePartitioningMethod_2019b, ananduta_EventtriggeredPartitioningNoncentralized_2021b}, which considers linear switching systems and event-driven systems. In these cases, conditions triggering the re-partition of the network are defined, together with ad-hoc strategies to perform it. The effects of partitioning on the properties of a control architecture are explored less frequently. An exception is \cite{baldiviesomonasterios_FeasibleSetsCoalitional_2019}, where feasibility regions and robust invariant sets are defined for distributed tube-based MPC.

\subsubsection{Mixed Strategies}
In \cite{chanfreut_FastClusteringMultiagent_2022} coalitional control, a bottom-up approach, is coordinated by a supervisory layer introducing a clustering problem for the selection of the size and composition of the coalitions, which is a top-down decision. Partitioning is achieved through a binary quadratic program where the decision variables are the states of the links between agents, and the cost is a composition of the gradient vectors of the control objective. However, such approach inherits the computation complexity of the general clustering problem, thus requiring the solution of an NP-hard problem.

\subsection{Preliminary Concepts and Notation} \label{sec:progoloma}
A graph \cite{bollobas_ModernGraphTheory_2001} is an ordered pair of sets $\mathcal{G}=(\mathcal{V},\mathcal{E})$ where $\mathcal{V} = \{1,\ldots, n\}$ is the set of $n$ vertices, and  $\mathcal{E}\subseteq\mathcal{V}\times\mathcal{V}$ is the set of the edges. The edges are associated to the vertices through a binary adjacency matrix $A_{\text{adj}}$, where $A_{\text{adj}, (i,j)} = 1$ if and only if an edge $\epsilon_{ij} = (i,j) \in\mathcal{E}$ exists. Therefore, the topology of the graph is specified by the adjacency matrix $A_{\text{adj}}$, and the set of the edges can also be written as $\mathcal{E} = \{(i,j)\,|\,i,j\in\mathcal{V}\wedge A_{\text{adj}, (i,j)} = 1\}$. A subgraph of $\mathcal{G}$ is a graph $\mathcal{S}_i=(\mathcal{V}_i,\mathcal{E}_i)$ representing a part of $\mathcal{G}$. The set of vertices $\mathcal{V}_i$ is a subset of $\mathcal{V}$, i.e.\ $\mathcal{V}_i\subseteq\mathcal{V}$, and the set of the edges is $\mathcal{E}_i = \{(i,j)\,|\,i,j\in\mathcal{V}_i\wedge A_{\text{adj}, (i,j)} = 1\}$, where the topology is still specified by the relevant entries of $A_{\text{adj}}$. For a directed graph $\mathcal{G}$, an edge $\epsilon_{ij} = (i,j)$ denotes an arrow starting from node $i$ and ending in node $j$. A graph $\mathcal{G}$ is weighted if a weighting matrix $W_{\text{adj}}$ assigning to each edge a number is specified in addition to $A_{\text{adj}}$. For each vertex $i\in\mathcal{V}$ we denote by $d_i$ its degree, i.e.\ the number of edges entering or exiting that vertex. In directed graphs, we can specify an in-degree ($d_{i,\text{in}}$) and an out-degree ($d_{i,\text{out}}$), if the edge is respectively ending or starting in the vertex $i$. For a vertex $i$, the neighborhood of $i$ is the set of all vertices connected to it, and we denoted it by $\mathcal{N}_i = \{j\in\mathcal{V} \,\,|\,\,(i,j)\vee(j,i)\in\mathcal{E}\}$. For a subgraph $\mathcal{S}_i=(\mathcal{V}_i,\mathcal{E}_i)$, the frontier is its set of nodes that are connected to nodes outside the subgraph, i.e.\ $\mathcal{F}_i = \{i\in\mathcal{V}_i\,\,|\,\,(i,j)\vee (j,i)\in\mathcal{E}, j\in\mathcal{V}\setminus\mathcal{V}_i\}$.

\section{The Equivalent Graph of a Dynamical System} \label{sec:associated_graph} 
In this section, we introduce the definition of the equivalent graph of a dynamical system for a general class of nonlinear discrete-time systems. Then, we characterize the definition for two important classes of dynamical systems: linear and piecewise affine (PWA) systems. For both classes, we highlight the properties that can be established from their topology, which are relevant for defining a partitioning strategy.

\subsection{General Definition of Equivalent Graph} \label{sub-sec:qeuivalent_graph}
Consider a discrete-time system of the form:
\begin{equation} \label{eq:state-dynamics}
	x(k+1) = f(x(k),u(k)) + g
\end{equation} 
with $x\in\mathcal{X}\subseteq\mathbb{R}^n$, $u\in\mathcal{U}\subseteq\mathbb{R}^p$, with $f$ differentiable over $\Omega = (\mathcal{X},\mathcal{U})$ almost everywhere (i.e.\ except for a Lebesgue set of measure zero), and $g$ is a constant. Moreover, assume that $f(0,0) = 0$. The equivalent graph representation of \eqref{eq:state-dynamics} is a time-dependent graph $\mathcal{G}_k=(\mathcal{V},\mathcal{E}_k, w_k, \tilde{g})$, where $\mathcal{V}$ is the set of the vertices, $\mathcal{E}_k\subseteq\mathcal{V}\times\mathcal{V}$ is the set of the edges, $w_k$ is a weighting function for the edges that implicitly characterizes the topology of the graph, and $\tilde{g}$ is a labeling function for the vertices. To construct the equivalent graph, first, we define the set of vertices: a new vertex is defined for each input and state variable of \eqref{eq:state-dynamics}. These vertices belong respectively to two subsets of vertices, which are $\mathcal{V}_u = \{u_1,\ldots,u_p\}$ and $\mathcal{V}_x = \{x_1,\ldots,x_n\}$. The set of vertices of the equivalent graph representation is thus the union $\mathcal{V} = \mathcal{V}_u \bigcup \mathcal{V}_x$. 
Then, we define the weighting function $w_k$, from which we can also extract the edges of the graph. This function $w_k$ assigns to each couple of vertices a function defined as:
\begin{equation}\label{eq:weight}
	w_k(i,j) = \left\{
	\begin{matrix}
		\frac{\partial f_j(k)}{\partial i} & \text{for} & i\in\mathcal{V}, j\in\mathcal{V}_x \\
		0 & \text{for} & i\in\mathcal{V}, j\in\mathcal{V}_u
	\end{matrix}
	\right.
\end{equation}
From $w_k$, it is possible to distinguish between the sets of input-to-state and state-to-state edges, which are respectively:
\begin{subequations}\label{eq:edges-definition}
	\begin{flalign}
		\mathcal{E}_{ux, k} = \left\{(i,j)\,|\,i\in\mathcal{V}_u,j\in\mathcal{V}_x,w_k(i,j)\neq 0\right\} \\
		\mathcal{E}_{xx, k} =\left\{(i,j)\,|\,i,j\in\mathcal{V}_x,w_k(i,j)\neq 0\right\} 
	\end{flalign}
\end{subequations}
The set of the edges of the equivalent graph is then obtained as the union $\mathcal{E}_k=\mathcal{E}_{ux, k}\bigcup\mathcal{E}_{xx, k}$. We will use an analogous notation for the weighting function referring to the above sets, using $w_{ux, k}$ and $w_{xx, k}$ respectively.
Finally, the labeling $\tilde{g}$ of the vertices is selected to coincide with the constants $g$ of \eqref{eq:state-dynamics} for state vertices, and zero for input vertices:
\begin{equation}
	\tilde{g}(i) = \left\{
	\begin{matrix}
		g_i & \text{for} & i\in\mathcal{V}_x \\
		0 \hfill & \text{for} & i\in\mathcal{V}_u
	\end{matrix}
	\right.
\end{equation}

The equivalent graph $\mathcal{G}$ thus obtained is a variable structure that reflects the nonlinear nature of system \eqref{eq:state-dynamics}. Specifically, for each pair $(x(k),u(k))$, the topology and the weighting of the graph may change, while the set of the vertices is time-independent. We now show how to build the equivalent graph for linear and piecewise linear systems.

\subsection{Equivalent Graph Equivalent of a Linear System}\label{sec:equivalent-linear}
Consider the system:
\begin{equation} \label{eq:linear-state-dynamics}
	x(k+1) = Ax(k) + Bu(k)
\end{equation} 
where $A$ and $B$ are matrices of appropriate dimensions. The structure of the graph $\mathcal{G}$ is now time-independent and presents no labeling. Consequently, the sets of the edges are:
\begin{subequations}\label{eq:linear-edges}
	\begin{flalign}
		\mathcal{E}_{ux} = \left\{(i,j)\,|\,i\in\mathcal{V}_u, j\in\mathcal{V}_x, B_{(j,i)} \neq 0 \right\}\\
		\mathcal{E}_{xx} = \left\{(i,j)\,|\,i, j\in\mathcal{V}_x, A_{(j,i)} \neq 0 \right\}
	\end{flalign}
\end{subequations}  

\subsection{Equivalent Graph of a Piecewise Affine Linear (PWA) System}
Consider a PWA system of the form:
\begin{equation}\label{eq:PWA-dynamics}
	x(k+1) = A^q x(k) + B^q u(k) + g^q \quad \text{for} \quad \begin{bmatrix}
		x(k)\\ u(k)
	\end{bmatrix}\in \mathcal{C}^q,
\end{equation}
for operational modes $q = 1,\ldots, N$, with $\mathcal{C}^1,\ldots, \mathcal{C}^N$ convex polyhedra in the input-state space with non-overlapping interiors \cite{heemels_EquivalenceHybridDynamical_2001}. For such systems, the weighting function $w_k$ \eqref{eq:weight}, the set of the edges $\mathcal{E}$, and the labeling $\tilde{g}$ can change according to the convex polyhedron $\mathcal{C}^q$. Considering the set of the edges $\mathcal{E}$, depending on the transition, the change can occur in $\mathcal{E}_u$, in $\mathcal{E}_x$, in both simultaneously, or not occur. According to this description, the sets of the edges \eqref{eq:edges-definition} for PWA systems assume the characterization depending on the current mode $q = q(x(k), u(k))$:
\begin{subequations}\label{eq:PWA-edges}
	\begin{flalign}
		\mathcal{E}_{ux}^q = \left\{(i,j)\,|\,i\in\mathcal{V}_u, j\in\mathcal{V}_x, B_{(j,i)}^q \neq 0 \right\}\\
		\mathcal{E}_{xx}^q = \left\{(i,j)\,|\,i, j\in\mathcal{V}_x, A_{(j,i)}^q \neq 0 \right\}
	\end{flalign}
\end{subequations} 
and the set of the edges of the equivalent graph is $\mathcal{E}^q = \mathcal{E}_{ux}^q\bigcup\mathcal{E}_{xx}^q$. Therefore, we have a finite set of at most $N$ topologies that can vary according to the operational mode.    

\subsection{Number of Distinct Topologies for a Dynamical System}
\label{sebsec:number-topologies}

The topology of the equivalent graph of a dynamical system is determined by the set of the edges $\mathcal{E}$. As stressed above, for linear systems this topology is static; thus a linear system only has one possible configuration for the edges for a given pair $(A,B)$ of state space matrices. For PWA systems, the possible number of topologies is equal to the number $N$ of different operating modes at most. However, some operating modes may share the same topology and only differ in the formulation of the dynamics. For a general dynamical system of the form \eqref{eq:state-dynamics}, we can define an upper bound on the number of distinct topologies:
\begin{lemma}\label{prop:number-of-topologies}
	The graph associated with system \eqref{eq:state-dynamics} has at most $2^{n(n+p)}$ distinct topologies. If self-edges are not considered, the graph has at most $2^{n(n+p-1)}$ distinct topologies.
\end{lemma}
\begin{proof}
	This statement is shown by considering the definition \eqref{eq:edges-definition}, where for each pair of vertices $(i,j)$, we are interested only in the cases where the derivative is zero or non-zero. Thus, we have two possible combinations for each of the $n$ functions and $p$ or $n$ variables, which provides the result. If self-edges are not considered, the statement holds without accounting for the edges $(i,i)\in\mathcal{E}$ with $i\in \mathcal{V}_x$.
\end{proof}

\section{The Concepts of Composite System Unit and Control Partition, and the Structure of the Generalized Partitioning Strategy} \label{sec:algorithm_atomic_agents} 
\subsection{Composite System Unit and Control Partition}
{ One of the scopes of a partitioning strategy applied to a dynamical system is to select subsystems that can be considered as individual control units for a specified application. In this section, we first define a special type of subsystem, the \textit{composite system unit}, which allows the definition of a generalized partitioning strategy. Subsystems of this type can be systematically represented as graphs. Then, we show how to aggregate these subsystems. Finally, we provide the general definition of partitioning based on the concepts above. First, we consider the abstract definition of subsystem:}
\begin{definition}[Subsystem] \label{def:subsystem}
	For a dynamical system of the form \eqref{eq:state-dynamics}, a subsystem is a group of state dynamics collected into a vector $x^{[i]}$ and whose behavior is defined by $x^{[i]}(k+1) = f^{[i]}(x(k),u(k)) + g^{[i]}$, with $x^{[i]}\in\mathcal{X}^{[i]}\subseteq\mathbb{R}^{n_i}$.
\end{definition}

This definition is not sufficiently specialized to define a generalized partitioning strategy. To this end, we introduce the definition of composite system unit, which allows the systematic selection of input-state vector pairs $(x^{[i]},u^{[i]})$. The choice of these pairs is not unique; however, our characterization is motivated by the scope of defining a generalized partitioning strategy.

\begin{definition}[Composite System Unit] \label{def:control-agent}
	A composite system unit (CSU) is a nonautonomous dynamical system whose inputs affect only the state dynamics of the CSU itself. All the dynamic relations the CSU has with other CSUs occur through dynamical coupling among the states, or direct input-output connections. For the dynamics \eqref{eq:state-dynamics}, a CSU is a subsystem of the form $x^{[i]}(k+1) = f^{[i]}(x(k),u^{[i]}(k)) + g^{[i]}$  with $x^{[i]}\in\mathcal{X}^{[i]}\subseteq\mathbb{R}^{n_i}$, $u^{[i]}\in\mathcal{U}^{[i]}\subseteq\mathbb{R}^{p_i}$, and such that $w(k,l) = 0$ for all $l\in\mathcal{V}_{i,x}$, $k\in\mathcal{V}_{j,u}$.
\end{definition}

From this definition, it follows immediately that a CSU is a subsystem in the sense of Definition \ref{def:subsystem}. We denote a CSU $i$ associated to the input-state pair $(x^{[i]},u^{[i]})$ with the symbol $\mathcal{S}_i$. The graph equivalent to the CSU in Definition \ref{def:control-agent} is constituted by the set of vertices $\mathcal{V}_i = \mathcal{V}_{i,x} \cup \mathcal{V}_{i,u} \subseteq \mathcal{V}$, by the weighting function $w(i,j)$ with $i,j\in\mathcal{V}_i$, by the set of edges $\mathcal{E}_{i, k} = \left\{(i,j)\,|\,i,j\in\mathcal{V}_i, w_{i, k}(i,j) \neq 0 \right\}$, and by the labeling $\tilde{g}_i$. This is the graph $\mathcal{S}_{i, k} = (\mathcal{V}_i,\mathcal{E}_{i, k},w_{i, k},\tilde{g}_i)$.

If we consider linear discrete-time time-invariant systems of the form \eqref{eq:linear-state-dynamics}, then from Definition \ref{def:control-agent} a CSU indexed by $i$ has the form:	
\begin{subequations}\label{eq:agent}
	\begin{align} 
		& x^{[i]}(k+1) = A^{[i]}x^{[i]}(k) + B^{[i]}u^{[i]}(k) + \omega^{[i]}(k)\\
		& \omega^{[i]}(k) = \sum_{j \in\mathcal{N}_i}A^{[ij]}x^{[j]}(k)
	\end{align}	
\end{subequations}
where $x\in\mathcal{X}^{[i]}\subseteq\mathbb{R}^{n_i}$, $u\in\mathcal{U}^{[i]}\subseteq\mathbb{R}^{p_i}$, and $A^{[i]}$, $B^{[i]}$, $A^{[ij]}$ are matrices of appropriate dimensions. The signal $\omega_k^{[i]}$ is the coupling effect that CSU $i$ experiences with its neighboring states $j\in\mathcal{N}_i$. 

In the definition of a general partitioning strategy, a fundamental task is the selection of the smallest possible CSUs constituting the network. We specify this concept through the definition of a fundamental system unit.
\begin{definition}[Fundamental System Unit]\label{def:atomic-control-agent}
	Given a network of dynamical systems, a fundamental system unit (FSU) is the smallest CSU definable through any network decomposition.
\end{definition}	

FSUs represent the smallest individual components that can be selected in a network without losing the nonautonomous property of the dynamics, i.e.\ the smallest subsystem that has at least one control input and one state variable. Before proceeding with the general definition of partitioning, we introduce the aggregation operation for two CSUs, which allows the merging of two distinct dynamics to form a bigger subsystem.

\begin{definition}[Aggregation Operation]\label{def:aggregation}
	Consider the equivalent graph representation $\mathcal{G}_k$ of the dynamical system \eqref{eq:state-dynamics}, and two subsystems of the same dynamical system described by the equivalent graphs $\mathcal{S}_{1, k}$, $\mathcal{S}_{2, k}$. We define the aggregation of the two subsystems over the graph $\mathcal{G}$ as $\mathcal{S}_{(1,2),k} = \left.\left(\mathcal{S}_{1, k}\uplus\mathcal{S}_{2, k}\right)\right\lvert_{\mathcal{G}_k}$:
	\begin{equation*}
		\mathcal{S}_{(1,2),k} = \left(\mathcal{V}_{1}\cup\mathcal{V}_{2},\mathcal{E}_{1, k}\cup\mathcal{E}_{2, k}\cup\mathcal{E}_{(1,2),k},w_{(1,2),k},(\tilde{g}_1,\tilde{g}_2)\right)
	\end{equation*}
	where $w_{(1,2),k} =\{w_k(i,j)\,|\,i,j \in \mathcal{V}_1\cup\mathcal{V}_2\}$, $\mathcal{E}_{(1,2),k} = \{(i,j)\,|\,((i\in\mathcal{V}_1,j\in\mathcal{V}_2)\vee (i\in\mathcal{V}_2,j\in\mathcal{V}_1))\wedge w_k(i,j)\neq 0\}$.
\end{definition}


The following proposition allows the construction of CSUs as the result of the aggregation of other CSUs: 
\begin{proposition} \label{prop:aggregation}
	The aggregation of multiple CSUs is a CSU. 
\end{proposition}	
\begin{proof}
	The statement can be verified by construction. 
\end{proof}

We conclude this section by providing the general definition of partitioning for a dynamical system:
\begin{definition}[Control Partition of a Dynamical System]
	For a dynamical system of the form \eqref{eq:state-dynamics} with equivalent graph $\mathcal{G}_k$, a control partition $\mathcal{P}_k$  is defined as a collection of CSUs:
	\begin{equation}
		\mathcal{P}_k = \{\mathcal{S}_{1,k},\ldots,\mathcal{S}_{m,k}\} \quad \text{s.t.} \quad \left.\left(\biguplus_{i=1}^m \mathcal{S}_{i,k}\right)\right\lvert_{\mathcal{G}_k} = \mathcal{G}_k.
	\end{equation}
\end{definition}
This definition of control partition comprehends both non-overlapping and overlapping partitionings thanks to the definition of aggregation operation.

\subsection{The Structure of the Generalized Partitioning Strategy}\label{sec:introduciton-partitioning}
In this section, we present our novel generalized partitioning technique\footnote{The optimization problem and algorithms, as well as the examples and case studies presented below, can be accessed through the repository \cite{riccardi_CodeUnderlyingPublication_2024}.}. This technique applies to all types of dynamical systems that can be represented in form \eqref{eq:state-dynamics}. Here, we describe how the algorithms, metrics, and optimization problem presented in the next sections merge to construct the generalized partitioning strategy. A generalized partitioning technique consists of two main steps:
\begin{enumerate}
	\item Selection of the FSUs, performed through an algorithmic procedure.
	\item Partitioning by aggregation, achieved either by solving an integer optimization problem, or through an algorithm. 
\end{enumerate}
For step 1), we describe the algorithmic approach in Section \ref{sec:algorithm-selection}. This algorithm is designed to work over the equivalent graph representation of the dynamical system.

Once the FSUs are defined, we are presented with a collection of subsystems that can not be further divided without losing the non-autonomicity characteristic. By aggregating these subsystems, we obtain the desired partition of the system. To this end, first, we introduce a novel metric in Section \ref{sec:partition-index}, which we call the partition index. This is a generalized global metric that needs information about the entire structure of the network and can be characterized according to the partitioning procedure that one wants to deploy. Then, we formulate the integer optimization problem to obtain the best value of the partition index in Sec.\ \ref{subsec:optimization-partitioning}. Alternatively, two algorithmic approaches to solve the same problem with a different characterization of the partition index are presented in Sec.\ \ref{subsec:algorithmic-partitioning} and \ref{subsec:improved-algorithmic-partitioning}. All strategies have advantages and disadvantages, which will be discussed in the respective sections. 

\section{Algorithm for the Selection of FSUs} \label{sec:algorithm-selection} 
{Here, we present the algorithm for the selection of the FSUs. This algorithm is general and can be applied to any network with structure \eqref{eq:state-dynamics}. Given the time-dependence of \eqref{eq:state-dynamics}, this algorithm is potentially applied at each time step $k$, but in the following we omit the subscript $k$ for simplicity. Also, the algorithm can be specialized to exploit the specific structure of the dynamics for linear and PWA systems. Before proceeding, we introduce some additional notation that will be used in the algorithm. We will indicate the FSUs as subsystems $\mathcal{A}_i$, for $i = 1,\ldots,N_{\textrm{FSU}}$, where $N_{\textrm{FSU}} \leq |\mathcal{U}|$,  and the collection of all the FSUs as $\mathcal{A} = \{\mathcal{A}_1,\ldots,\mathcal{A}_{N_{\textrm{FSU}}}\}$.  Also, we will use the additional set $\mathcal{L}\subset\mathcal{X}$, indicating the state nodes that remain to be assigned to the FSUs after preliminary operations are complete.
	
The general algorithm for selecting FSUs consists of three main steps that we present in the following in detail. These steps will be performed in a specific order as a part of the iterative procedure illustrated in Alg.\ \ref{alg:atomic}:
\begin{enumerate}
	\item \textit{Selection of the roots of the FSUs}.  
	To select the FSUs, we first select their anchor points, one for each FSU. We call these anchor points the roots of the FSUs. These roots consist of at least one input node and one state node directly connected by an edge. The selection of the roots is performed on the subgraph $(\mathcal{V},\mathcal{E}_{ux},w_{ux},\tilde{g})$, and works as follows. We start by assuming to have one FSU $\mathcal{A}_i\in\mathcal{A}$ for each input node $i\in\mathcal{U}$, where the indexing of the FSUs in the collection $\mathcal{A}$ may vary as the procedure progresses. Then, we consider for each input node in $\mathcal{U}$ the set of the edges connecting them to the state nodes: for $i\in\mathcal{U}$, if $w_{ij}\neq 0$, with $j\in\mathcal{X}$, we assign the state node $j$ to the FSU $\mathcal{A}_l$ to which the input node $i$ belongs. If we encounter a state node $j\in\mathcal{X}$ connected to two or more input nodes in $\mathcal{U}$, the entire group constitutes an indivisible root, and all are merged into the same FSU. Consequently, we update the indexing of the FSUs in the collection $\mathcal{A}$. We can have a maximum number of roots, and therefore of FSUs, equal to $N_{\textrm{FSU}} \leq |\mathcal{U}|$. All the unassigned state nodes remaining from the selection of the roots are collected into the list $\mathcal{L}$ that will be used in the next steps.
	\item \textit{Forward assignment of the state nodes}. In this step, we perform the assignment of state nodes not belonging to the roots of the FSUs already, and for which a directed edge from the state nodes in the roots exists. We call this procedure \textit{forward assignment of the state nodes}. Here we consider the graph $(\mathcal{X},\mathcal{E}_{xx}, w_{xx}, \tilde{g})$. For each unassigned node $j\in\mathcal{L}$, we consider the edge $w_{ij}$ starting from a state node $i\in\mathcal{X}$ in the roots and ending in $j$ that has the largest absolute value, i.e.\ $\max_{i}|w_{ij}|$. If this $w_{ij}$ exists, we assign the node $j\in\mathcal{L}$ to the FSU to which the state node $i=\arg\max_{i}|w_{ij}|$ belongs, and we update the set $\mathcal{L}$. If such $w_{ij}$ does not exist, we proceed with the assignment procedure for the next state node in $\mathcal{L}$ until no further forward assignments are possible. We call this phase \textit{forward assignment} since we look at the edges starting from the states in the roots of the FSUs and that are directed toward the periphery of the equivalent graph.
	\item \textit{Backward assignment of the state nodes}. This step of the algorithm assigns the state nodes in $\mathcal{L}$ that have no forward connection from the roots. Therefore, we call this step \textit{backward assignment} because we consider the connections from the nodes belonging to the periphery of the graph toward the direction of the roots. For each remaining $i\in\mathcal{L}$, following the same principle in step 2), we select the node $j$ belonging to an FSU such that $j=\arg\max_j |w_{ij}|$. We repeat the procedure for the next state node in $\mathcal{L}$ until we have checked all the nodes once. 
\end{enumerate}
\begin{algorithm}
	\DontPrintSemicolon
	\caption{Selection of FSUs}\label{alg:atomic}
	\textbf{Part 1 - Selection of the roots}\;
	\KwData{$\left(\mathcal{V},\mathcal{E}_{ux},w_{ux},\tilde{g}\right)$}
	Perform 1) in Sec.\ \ref{sec:algorithm-selection}: selection of the roots\;
	return $\left(\mathcal{A}^{[0]},\mathcal{L}^{[0]}\right)$\;
	\textbf{Part 2 - Selection of the FSUs}\;
	\KwData{$\left(\mathcal{X},\mathcal{E}_{xx},w_{xx},\tilde{g}\right)$}
	$k$ $\gets$ 0\;
	\While{$\mathcal{L}^{[k]}$ not empty $\vee$ $\mathcal{L}^{[k]}$ connected}{
		Perform 2) in Sec.\ \ref{sec:algorithm-selection}: forward assignment\;
		$\left(\mathcal{A}^{[k+1]},\mathcal{L}^{[k+1]}\right)$ $\gets$ $\left(\mathcal{A}^{[k]},\mathcal{L}^{[k]}\right)$\;
		Perform 3) in Sec.\ \ref{sec:algorithm-selection}: backward assignment\;
		$\left(\mathcal{A}^{[k+2]},\mathcal{L}^{[k+2]}\right)$ $\gets$ $\left(\mathcal{A}^{[k+1]},\mathcal{L}^{[k+1]}\right)$ \;
		$k$ $\gets$ $k+2$
	}
\end{algorithm}
	
\begin{remark}
	If the system considered is time-varying, its topology can change at each time step. For such systems, the algorithm for selecting the FSUs has to be applied each time there is a variation in the topological structure, i.e.\ in the edges or their weights.
\end{remark}
\begin{remark}
	If outputs are also specified in the dynamics of the system, the procedure defined above still holds, but it is necessary to assign the output nodes to each CSU in the same way as done for the state nodes. 
\end{remark}

\section{The Partitioning Strategy}\label{sec:partitioning-strategy} 
In this section, we present a novel metric that will be used to address the partitioning problem in a generalized setting. We call this metric the \textit{partition index}. Then, we propose an algorithmic procedure to solve the problem, which returns the partition, i.e.\ a set of CSUs, for the required granularity. This approach provides a sub-optimal solution but requires polynomial time to be computed. In Sec.\ \ref{subsec-partitioning_modular}, we empirically show how this first approach offers a good trade-off between optimality of the solution and computational complexity.
Moreover, we show that when the partition index is a quadratic function, the partitioning problem can be formulated as an integer quadratic program (IQP). This will require a novel characterization of the partition index, which must be a quadratic function. In Sec.\ \ref{subsec-partitioning_modular}, we show how this approach provides the best solution to the problem, but the downside is that it is an NP-hard problem.

\subsection{The partition index}\label{sec:partition-index}

	For a specific partition $\mathcal{P}$ we denote the \textit{partition index} by $p^{\text{idx}}(\mathcal{P})$. We define this metric to account for both the interactions that occur inside a CSU, and among the CSUs. In this sense, the partition index is a global metric since it requires global information about the system to be computed. 
	
	Consider a partition $\mathcal{P} = \{\mathcal{S}_1,\ldots,\mathcal{S}_{N_{\text{CSU}}}\}$, with $N_{\text{CSU}}\leq N_{\textrm{FSU}}\leq p$ where each $\mathcal{S}_i$ is a CSU in the sense of Definition \ref{def:control-agent}. For CSU $\mathcal{S}_i$ we define following measures: the intra-CSU interaction $W^{\text{intra}}_{\mathcal{S}_i}$, the inter-CSU interaction $W^{\text{inter}}_{\mathcal{S}_i}$, and the size of the CSU $W^{\text{size}}_{\mathcal{S}_i}$, i.e.\ the number of FSUs belonging to $\mathcal{S}_i$. The components $ W_{\mathcal{S}_i}^{\text{inter}}, W_{\mathcal{S}_i}^{\text{intra}},  W_{\mathcal{S}_i}^{\text{size}}$ of the partition index \eqref{eq:partition_index} can be constructed through the equivalent graph representation. For this, consider a non-linear system of the form \eqref{eq:state-dynamics}, with equivalent graph representation $\mathcal{G}=(\mathcal{V},\mathcal{E}, w, \tilde{g})$ where the dependence on $k$ is omitted. A CSU $\mathcal{S}_i$ belonging to such system has an equivalent graph $\mathcal{S}_i=(\mathcal{V}_i,\mathcal{E}_i, w_i, \tilde{g}_i)$.
	For CSU $\mathcal{S}_i$ we can define the intra-CSU, and the inter-CSU interactions respectively as
	\begin{align}
		& W_{\mathcal{S}_i}^{\text{intra}} = \sum_{s,t\in\mathcal{V}_{i}} |w_{i}(s,t)| \label{eq:intra-agent-nonlinear} \\
		& W_{\mathcal{S}_i}^{\text{inter}} = \sum_{s\in\mathcal{F}_{\mathcal{S}_i}} \sum_{j\in\mathcal{N}_{\mathcal{S}_i}} \sum_{t\in\mathcal{N}_{s}\cap\mathcal{V}_j} |w_{i}(s,t)| + |w_{j}(t,s)|  \label{eq:inter}
	\end{align}
	where $\mathcal{F}_{\mathcal{S}_i}$ denotes the frontier of the CSU $\mathcal{S}_i$, $\mathcal{N}_{\mathcal{S}_i}$ the set of neighbor CSUs of $\mathcal{S}_i$, and $\mathcal{N}_{s}$ the set of neighbor nodes of node $s$. 
	Note that $W_{\mathcal{S}_i}^{\text{size}}$ is equal to the square of the number of FSUs $\mathcal{A}$ belonging to the CSU $\mathcal{S}_i$, i.e.\ $W_{\mathcal{S}_i}^{\text{size}} = |\mathcal{V}_{i}|^2$.
	The partition index is defined as:
	\begin{definition}[Partition Index]
		The partition index of a network of CSUs for a specified set-valued function $h$ and a given partition $\mathcal{P} = \{\mathcal{S}_1,\ldots,\mathcal{S}_m\}$ is a global metric given by:
		\begin{equation}\label{eq:partition_index}
			p^{\textnormal{idx}}(\mathcal{P}) = h\left(\sum_{i=1}^{m}W_{\mathcal{S}_i}^{\textnormal{inter}}, \sum_{i=1}^{m} W_{\mathcal{S}_i}^{\textnormal{intra}}, \sum_{i=1}^{m} W_{\mathcal{S}_i}^{\textnormal{size}}, \alpha\right)
		\end{equation}	
		where $\alpha>0$ is a scalar parameter allowing for the selection of the granularity of the partitioning
	\end{definition}
	In the general case of the nonlinear dynamics \eqref{eq:state-dynamics}, $p^{\textnormal{idx}}(\mathcal{P})$ is time-dependent, because is is based on a time varying graph. The framework for partitioning that we propose is based on the optimization of $p^{\textnormal{idx}}(\mathcal{P})$, and works for every characterization of $h$. In the following we propose one nonlinear form for $h$ that we will use for algorithmic partitioning in Sec.\ \ref{subsec:algorithmic-partitioning}; and one quadratic for the optimization-based strategy in Sec.\ \ref{subsec:optimization-partitioning}. Other strategies using different characterizations of $h$ can be defined by the same approach.
	
	\subsection{Algorithmic partitioning}\label{subsec:algorithmic-partitioning}  
	For the first characterization of the partition index \eqref{eq:partition_index}, we assume a function that accounts for the ratio between the intra- and inter-CSU interactions. Since the algorithmic approach will try to maximize the partition index, then we assume that terms $W_{\mathcal{S}_i}^{\text{intra}}$ are at the numerator, and $W_{\mathcal{S}_i}^{\text{inter}}$ at the denominator. This approach will favor merging strongly coupled FSUs into the same CSU, while assigning weakly coupled FSUs to different CSUs. This approach is inspired by the modularity metric \cite{newman_ModularityCommunityStructure_2006a, brandes_MaximizingModularityHard_2006}, and the motivation for using it in non-centralized control is that such approaches modular partitioning usually improve performances of the architecture by minimizing the interaction among subsystem. Also, advantages in communication costs and privacy are present. However, it is important to stress that this partitioning approach will not provide the optimal partitioning in absolute terms (as we will show in Sec.\ \ref{sec:example-case-study}), and which motivates the introduction of the abstract formulation \eqref{eq:partition_index}. Moreover, we include in the \eqref{eq:partition_index} a novel term that accounts for the size of the resulting CSUs, according to the parameter $\alpha$ allowing for selecting the desired granularity of the partitioning: 
	
	\begin{equation} \label{eq:metric_greedy}
		p^{\textrm{idx}}(\mathcal{P}) = \frac{\displaystyle  \sum_{i=1}^{m}W_{\mathcal{S}_i}^{\text{Intra}}}{\displaystyle 1 + \sum_{i=1}^{m}W_{\mathcal{S}_i}^{\text{Inter}}} + \frac{\alpha}{\displaystyle 1 + \sum_{i=1}^{m} W_{\mathcal{S}_i}^{\text{size}}}
	\end{equation}
	This index grows when the interactions among FSUs in the same CSU increase, the interactions among CSUs decrease, and the sizes of the CSUs decrease. A greedy algorithm maximizing this index is reported in Alg.\ \ref{alg:greedy}. The time complexity of this algorithm is $O(N_{\textrm{FSU}}^3)$ because for each of the $N_{\textrm{FSU}}$ iterations it is necessary to evaluate the immediate gain in \eqref{eq:metric_greedy} of the assignment of every unassigned FSU to each CSU, which can be at most $N_{\textrm{FSU}}$, and this operation has to be performed $N_{\textrm{FSU}}$ times, until no more FSUs have to be assigned.
	
	In the algorithm, we start by creating a list of unassigned FSUs $\mathcal{L}$, and an empty partition $\mathcal{P}$, which is a list with $N_{\textrm{FSU}}$ entries, where the $i$-th entry is itself a sub-list denote by $\mathcal{P}_i$, also initialized as empty. The partition index corresponding to the empty allocation at the first step is $p^{\textrm{idx}}(\mathcal{P}) = 0$. At each iteration, the algorithm tries all possible assignments of the elements in $\mathcal{L}$ to all the sub-lists $\mathcal{P}_i$. For each assignment, it evaluates which element of $\mathcal{L}$ maximizes the immediate return given by the difference in \eqref{eq:metric_greedy} evaluated for $\mathcal{P}$ and for $\mathcal{P}^{\text{New}}$, where the latter is the candidate next partitioning under evaluation for an element of $\mathcal{L}$. This operation provides the maximum immediate gain $\Delta p^{\textrm{max}}$, which is re-initialized to minus infinity at each iteration of the while-loop. If the assignment of FSU $\mathcal{A}_i$ to the sub-list $\mathcal{P}_j$ gives a $\Delta p^{\textrm{idx}} > \Delta p^{\textrm{max}}$, then we store the candidate best allocation in the variables $\mathcal{A}_{\textrm{Next}}$, $j_{\textrm{Next}}$. When all possible evaluations in $\mathcal{L}$ are performed, the resulting $\mathcal{A}_{\textrm{Next}}$ is assigned to $\mathcal{P}_{j_{\textrm{Next}}}$, and removed from $\mathcal{L}$. The algorithm stops when $\mathcal{L}$ is empty, returning the partition $\mathcal{P}$.  
	
	
	\SetKwComment{Comment}{$\#$ }{ }
	
	\begin{algorithm}
		\DontPrintSemicolon
		\caption{Greedy algorithm for partitioning}\label{alg:greedy}
		\KwData{$\mathcal{A} = \{\mathcal{A}_1,\ldots,\mathcal{A}_{N_{\textrm{FSU}}}\}$}
		\KwResult{$\mathcal{P}$}
		$\mathcal{L}$ $\gets$ $\mathcal{A}$, $\mathcal{P}$ $\gets$ $\{\mathcal{P}_1,\ldots,\mathcal{P}_{N_{\textrm{FSU}}}\}$, $p^{\textrm{idx}}(\mathcal{P})$ $\gets$ 0\;
		\While{$length(\mathcal{L})$ $ > 0$}{
			$\Delta p^{\textrm{max}}$ $\gets$ $-\infty$, $\mathcal{A}_{\textrm{Next}}$ $\gets$ None, $j_{\textrm{Next}}$ $\gets$ None\;
			\For{$i$ = $1,\ldots,length(\mathcal{L})$}{
				$\mathcal{A}_i$ $\gets$ $\mathcal{L}_i$\;
				\For{$j$ = $1,\ldots,length(\mathcal{P})$}{
					$\mathcal{P}^{\text{New}}$ $\gets$ $\mathcal{P}$ s.t. $\mathcal{P}_j$.append($\mathcal{A}_i$)\;
					$\Delta p^{\textrm{idx}}$ $\gets$ $p^{\textrm{idx}}(\mathcal{P}^{\text{New}}) - p^{\textrm{idx}}(\mathcal{P})$\;
					\If{$\Delta p^{\textnormal{idx}}$ $>$ $\Delta p^{\textnormal{max}}$}{
						$\mathcal{A}_{\textrm{Next}}$ $\gets$ $\mathcal{A}_i$, $j_{\textrm{Next}}$ $\gets$ j, $\Delta p^{\textrm{max}}$ $\gets$ $\Delta p^{\textrm{idx}}$\;
					}
				}
			}
			$\mathcal{P}$ $\gets$ $\mathcal{P}$ s.t. $\mathcal{P}_{j_{\textrm{Next}}}$.append($\mathcal{A}_{\textrm{Next}}$), $\mathcal{L}$.remove($\mathcal{A}_{\textrm{Next}}$)\;
		}
	\end{algorithm}

{\subsection{Improved algorithmic partitioning}\label{subsec:improved-algorithmic-partitioning}   
	We improve the partition quality returned by the greedy Algorithm \ref{alg:greedy} by introducing a local refinement step at the expense of higher computational complexity. Specifically, once the assignment of an FSU to a CSU is performed, we can run a nested evaluation of the partition, going through all the FSUs already allocated and checking whether relocating them to a different CSU improves the value of the partition index, and iterating this procedure until no further improving by swapping is possible. A similar idea was proposed in \cite{blondel_FastUnfoldingCommunities_2008a}, where an analogous procedure constitutes the overall partitioning strategy.
	The local refinement routine is reported in Alg.\ \ref{alg:greedy-improved}, and the improved algorithm is obtained by performing the steps in Alg.\ \ref{alg:greedy-improved} right after the last step in the while loop of Algorithm \ref{alg:greedy}, i.e.\ after $\mathcal{L}$.remove($\mathcal{A}_{\textrm{Next}}$) in line 11. Since in the worst case Alg.\ \ref{alg:greedy-improved} re-runs three times through all the FSUs in a while loop, the overall computational complexity of the improved algorithm is $O(N_{\textrm{FSU}}^4)$.
	\begin{algorithm}
		\DontPrintSemicolon
		\caption{Partition refinement}\label{alg:greedy-improved}
		\KwData{$\mathcal{P}$}
		\KwResult{$\mathcal{P}$}
		$\Delta p^{\textrm{max}}$ $\gets$ $0$, Exit\_Condition $\gets$ False \;
		\While{Exit\_Condition $=$ False}{
			$\mathcal{A}_{\textrm{Next}}$ $\gets$ None, $j_{\textrm{Next}}$ $\gets$ None, $j_{\textrm{Prev}}$ $\gets$ None\;
			\For{$i$ = $1,\ldots,length(\mathcal{P})$}{
				\For{$j$ = $1,\ldots,length(\mathcal{P}_i)$}{
					$\mathcal{A}_{\text{Test}}$ $\gets$ $\mathcal{P}_{i,j}$ \; $\mathcal{P}^{\text{New}}$ $\gets$ $\mathcal{P}$ s.t. $\mathcal{P}_i$.remove($\mathcal{A}_{\text{Test}}$)\;
					\For{$k$ = $1,\ldots,length(\mathcal{P})$}{
						\If{k != i}{
							$\mathcal{P}^{\text{New}}$ $\gets$ $\mathcal{P}^{\text{New}}$ s.t. $\mathcal{P}_k$.append($\mathcal{A}_{\text{Test}}$)\;
							$\Delta p^{\textrm{idx}}$ $\gets$ $p^{\textrm{idx}}(\mathcal{P}^{\text{New}}) - p^{\textrm{idx}}(\mathcal{P})$\;
							\If{$\Delta p^{\textrm{idx}} > \Delta p^{\textrm{max}}$}{
								$\mathcal{A}_{\textrm{Next}}$ $\gets$ $\mathcal{A}_{\textrm{Test}}$\;
								$j_{\textrm{Next}}$ $\gets$ $k$, $j_{\textrm{Prev}}$ $\gets$ $i$ \;
							}
						}
						
					}
				}
			}
			\If{$\mathcal{A}_{\textrm{Next}}$ != None}{
				$\mathcal{P}$ $\gets$ $\mathcal{P}$ s.t. $\mathcal{P}_{j_{\textrm{Prev}}}$.remove($\mathcal{A}_{\text{Next}}$)\;
				$\mathcal{P}$ $\gets$ $\mathcal{P}$ s.t. $\mathcal{P}_{j_{\textrm{Next}}}$.append($\mathcal{A}_{\text{Next}}$)\;
			}\Else{
				Exit\_Condition $\gets$ True\;
			}
			
		}
	\end{algorithm}
}

\subsection{Optimization-based partitioning}\label{subsec:optimization-partitioning}   
The characterization of the partition index proposed in \eqref{eq:metric_greedy} can be interpreted as a nonlinear mixed-integer metric, for which, in general, it is not possible to formulate programming problems allowing for a global optimum. To address this limitation, in this section we propose a characterization of the partition index \eqref{eq:partition_index} that is a mixed-integer quadratic function. With this choice, we will define a multi-objective optimization problem that will try to balance the difference between $W_{\mathcal{S}_i}^{\text{inter}}$ and $W_{\mathcal{S}_i}^{\text{intra}}$, while accounting for the size of the resulting CSUs through the product of $\alpha$ and $W_{\mathcal{S}_i}^{\text{size}}$. The approach is conceptually similar to the one used to derive \eqref{eq:metric_greedy}, but provides a mixed-integer quadratic metric. Consequently, the optimization-based strategy for partitioning is based on the solution of an Integer Quadratic Programming (IQP) problem, for which an optimal solution can be obtained.

Under these considerations, we introduce the binary variables $\delta_{ij}\in\{0,1\}$ to specify whether or not an FSU $\mathcal{A}_i$ belongs to a CSU $\mathcal{S}_j$:
\begin{equation}
	\delta_{ij} = 1  \Longleftrightarrow \mathcal{A}_i \in \mathcal{S}_j
\end{equation}
All decision variables $\delta_{ij}$ are collected into the vector $\delta$. The total number of entries of $\delta$ is $N_{\textrm{FSU}}^2\leq p^2$, since the range of $i,j$ is $1,\ldots,N_{\textrm{FSU}}$. The number of non-empty CSUs $\mathcal{S}_i$ that the optimization procedure will generate is unknown a priori, but it cannot exceed the number of FSUs by definition.
Using the binary variables $\delta$, we can rewrite the terms in \eqref{eq:inter}-\eqref{eq:partition_index} respectively as:

\setlength{\arraycolsep}{0.0em}
\begin{eqnarray} \label{eq:IQP_metric_1}
	W^{\text{inter}}(\delta) = \sum_{m = 1}^{N_{\textrm{FSU}}} \sum_{i = 1}^{N_{\textrm{FSU}}}\sum_{\substack{j = 1\\ j\neq i}}^{N_{\textrm{FSU}}}\sum_{\substack{l = 1\\ l\neq m}}^{N_{\textrm{FSU}}}\delta_{i,m}\delta_{j,l}\left(|w(i,j)| \right. \quad \nonumber \\ \left. + |w(j,i)|\right)  \quad
\end{eqnarray}
\begin{eqnarray}\label{eq:IQP_metric_2}
	W^{\text{intra}}(\delta) = \sum_{m = 1}^{N_{\textrm{FSU}}} \sum_{i = 1}^{N_{\textrm{FSU}}}\sum_{j = 1}^{N_{\textrm{FSU}}}\delta_{i,m}\delta_{j,m}\left(|w(i,i)|\right. \nonumber \\ +|w(i,j)|+  
	\left.|w(j,i)|+  |w(j,j)|\right) 
\end{eqnarray}
\setlength{\arraycolsep}{5pt}
\begin{equation}\label{eq:IQP_metric_3}
	W^{\text{size}}(\delta) = \sum_{m = 1}^{N_{\textrm{FSU}}} \left(\sum_{i = 1}^{N_{\textrm{FSU}}} \delta_{i,m}\right)^2  
\end{equation}
The network partition is thus expressed as a function of vector $\delta$, i.e.\  $\mathcal{P}(\delta)$. Additionally, we include a set of constraints on assigning the FSUs to the CSUs so that we get a non-overlapping partitioning. This requirement is codified by imposing that the sum of variables $\delta_{i,j}$ over the index $j$ is equal to one, i.e.\ each FSU must belong only to one CSU. The resulting IQP problem is:
\begin{flalign}\label{eq:IQP}
	\min_\delta \quad &  p^{\textrm{idx}}(\mathcal{P}(\delta)) = W^{\text{inter}}(\delta) - W^{\text{intra}}(\delta) +\alpha  W^{\text{size}}(\delta) \nonumber \\
	\textrm{s.t.} \quad  &\sum_{j= 1}^{N_{\textrm{FSU}}}\delta_{ij} = 1  \qquad \forall i \\ 
	& \delta_{ij} \in \{0,1\} \hfill \nonumber
\end{flalign}
Solving the IQP problem \eqref{eq:IQP} provides the partition minimizing the partition index \eqref{eq:partition_index} for a given choice of the granularity parameter $\alpha$. Such optimization problems are known to be NP-hard \cite{conforti_IntegerProgramming_2014}; therefore, their scalability is limited, as will be shown in the examples in Sec.\ \ref{sec:example-selection}.

\begin{remark}
	For the problem definition \eqref{eq:IQP}, and for the algorithmic approach in Sec.\ \ref{subsec:algorithmic-partitioning}, it is possible to find a range for the selection of $\alpha$. In this range, for large values of $\alpha$, the resulting partitioning will be the set of individual FSUs, i.e.\ $\mathcal{P} = \{\mathcal{A}_1,\ldots,\mathcal{A}_{N_{\textrm{FSU}}}\}$. Conversely, for small values of $\alpha$, the resulting partitioning will be a single CSUs comprehending all FSUs, which is the entire system, i.e.\ $\mathcal{P} = \{\mathcal{S}_1\}$ with $\mathcal{S}_1 = \{\mathcal{A}_1,\ldots,\mathcal{A}_{N_{\textrm{FSU}}}\}$. Other choices of $\alpha$ in the range will lead to different levels of aggregation of the FSUs.
\end{remark}

\section{Examples: Application of the Generalized Partitioning Strategy}\label{sec:example-network-partitioning} 
\subsection{Selection of the FSUs}\label{sec:example-selection}


\begin{figure}[t]
	\centering
	\subfloat[]{
		\includegraphics[width=0.75\linewidth]{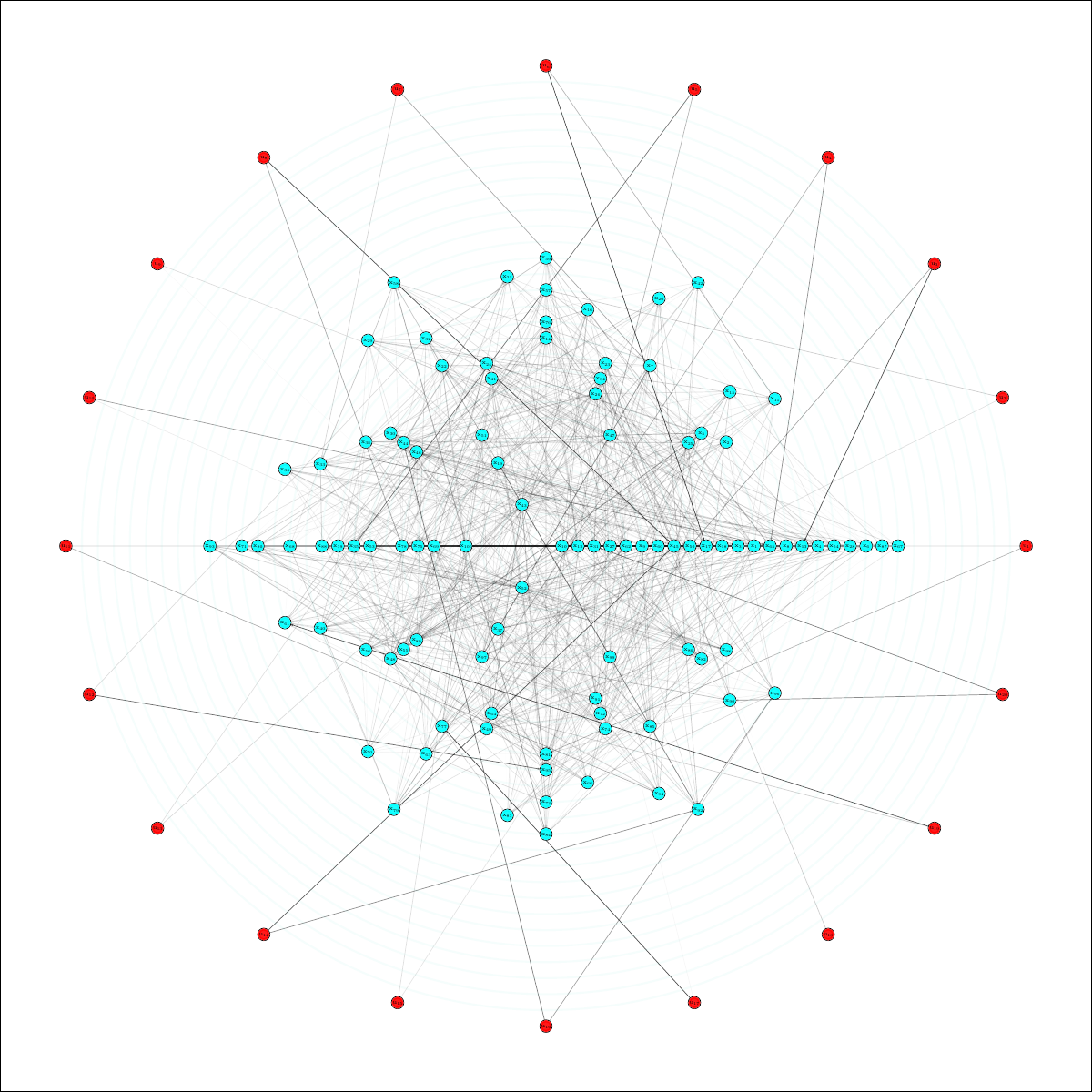}
		\label{fig:Network_States_100_Inputs_20network_map}
	}
	\hfill
	\subfloat[]{
		\includegraphics[width=0.75\linewidth]{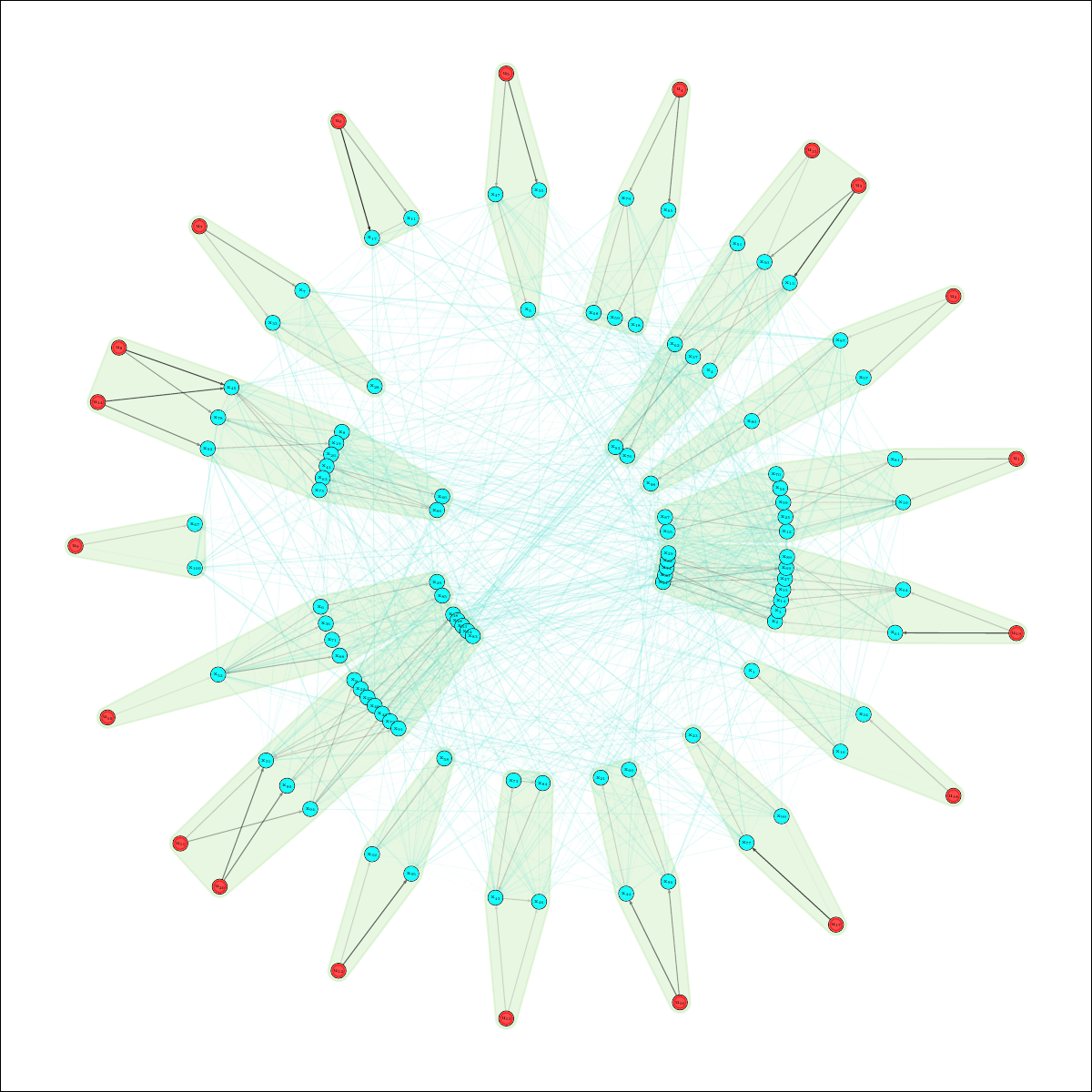}
		\label{fig:Network_States_100_Inputs_20atomic_control_agents}
	}
	\caption{Representations of the equivalent graph of a system with 20 input (in red) and 100 states (in cyan) variables. In (Fig.\ \ref{fig:Network_States_100_Inputs_20network_map}) the concentric degree-based representation, where input nodes are positioned in the most outer part of a circle, while state nodes on concentric rings representing their degree, with the most outer being degree one, and growing of one each ring toward the center. In (Fig.\ \ref{fig:Network_States_100_Inputs_20atomic_control_agents}) the representation of the FSUs after the application of the selection procedure. Each of the 17 FSU is represented by the green area, where dynamic relationships among the variables in the same FSU are represented by with black arrows, and among the FSUs by green arrows. }
\end{figure}


To show the features of Algorithm \ref{alg:atomic}, we consider a dynamical system with 100 states and 20 inputs.  
The system has an equivalent graph represented in Fig.\ \ref{fig:Network_States_100_Inputs_20network_map}. This graph can be interpreted as a snapshot of the dynamical relationships occurring through system variables at a certain time instant. The complexity of such graph makes the expert selection of FSUs, i.e.\ FSUs, a long and challenging task that quickly becomes intractable as the group size grows. In practice, no clear structure can be identified in such graph at first sight.
The application of Algorithm \ref{alg:atomic} reveals the existence of 17 FSUs, as shown in Fig.\ \ref{fig:Network_States_100_Inputs_20atomic_control_agents}. For such FSUamental units, all the couplings with their neighbors are guaranteed to occur through dynamic relationships, and no inputs are shared among them. As shown in Fig.\ \ref{fig:Network_States_100_Inputs_20atomic_control_agents}, this property might require the aggregation of some input nodes together to form single FSUs. It is interesting to notice that the existence of a forward path from an input to a state node inside an FSU implies, for discrete-time systems, that the input will affect that state in a number of time steps equal to the length of the path (an analogous consideration is possible for continuous-time systems in the number of their time derivatives \cite{daoutidis_StructuralEvaluationControl_1992}). Also, certain states, or groups of them, only exhibit a backward effect on the FSU they belong to. For such states, no forward path exists from any input node, making them unreachable states. Still, we assign them to their closest FSU according to their coupling strength.

\subsection{Partitioning a Modular Network} \label{subsec-partitioning_modular}
{ We now show the application of our partitioning strategy to a modular network \cite{newman_ModularityCommunityStructure_2006a}\cite{brandes_ModularityClustering_2008a} of FSUs reported in Fig.\ \ref{fig:Modular_64}. Such network is characterized by a basic pattern of four FSUs strongly coupled, but weakly coupled with similarly configured groups. This type of connection structure is repeated throughout the network at different levels of aggregation. We propose this example to show three important aspects of the proposed approach: 1) the use of the partition index \eqref{eq:partition_index} allows to obtain partitions that aggregate together strongly connected components of the network while minimizing the coupling with the others; 2) the function of the granularity parameter $\alpha$ defined in \eqref{eq:partition_index}, which allows selecting the relevance of the level of aggregation of the components of the partition, thus returning different configurations; 3) the respective advantages and limitations of the optimization-based and algorithmic partitioning approaches. 
	
	In the network in Fig.\ \ref{fig:Modular_64}, a basic pattern of four FSUs that are strictly connected can easily identified by inspection. This pattern repeats at a coarser level if we consider a group of 16 FSUs connected by their corner links. Further, the structure repeats if we consider the entire network of 64 FSUs. 
	To partition this network, we first apply the optimization-based strategy of Section \ref{subsec:optimization-partitioning}. To this end, we need to select a value for the granularity parameter $\alpha$. In our testing, we used a value related to the weights of the edges of the graphs, i.e.\ $\alpha = \left(\frac{\kappa}{w^{\text{min}}}\right)^2$, where $\kappa$ is a scalar free to select by the user, and $w^{\text{min}}$ is the minimum weight of the edges of the graph. By varying $\kappa$, we obtain different $\alpha$ and consequently different partitions. We tried a broad spectrum of values $\alpha$; however, we retrieved only four distinct partitions, which actually cover all the possible aggregations of the FSUs in Fig.\ \ref{fig:Modular_64} at different granularities that we can identify by inspection. Specifically, for $\kappa = 1, 0.1, 0.01, 0.001$ we have $\alpha = 10^6, 10^4, 10^2, 1$. The partitions obtained, which we denoted respectively as partitions $\mathcal{P}_1, \ldots, \mathcal{P}_4$, are represented in the corresponding Figs.\ \ref{fig:Partitioning01}-\ref{fig:Partitioning04}. As expected, the larger the value of $\alpha$ is, the higher the relevance of the size of the resulting CSU in the partition. If the value of $\alpha$ becomes sufficiently large, then the resulting partition is constituted by all individual FSUs. Conversely, if $\alpha$ is small enough, the network partition is one individual CSU aggregating all the FSUs together. The other two selections of $\alpha$ return groups of 4 or 16 FSUs, representing the other two optimal aggregations of FSUs.  
	When applying algorithmic partitioning to the same problem, we have retrieved two of the four optimal results obtained with optimization-based partitioning.  
	In particular, using the metric \eqref{eq:metric_greedy}, with a choice of $\alpha = 25$ we get the partition in Fig.\ \ref{fig:Partitioning01}, where every FSU is allocated into a distinct CSU. By reducing $\alpha$ to 1, we get the partition in Fig. \ref{fig:Partitioning01}, with 16 CSUs containing 4 FSUs each. For $\alpha$ smaller than 1, we obtained various combinations of the basic groups of four FSUs, but never the exact partition in Fig.\ \ref{fig:Partitioning03}
	
	The last aspect to consider in this example is the computational cost of the two approaches. We solved the optimization-based problem using the Gurobi Optimizer software \cite{GurobiOptimization}, running with parallel computing on a CPU Intel XEON E5-6248R with 24 cores at 3GHz. In the worst case scenario, that is for the partition in Fig. \ref{fig:Partitioning01} with 64 individual FSUs, we let the optimizer run for 3 hours and 34 minutes before prematurely stopping it while reaching a solution gap of $7.07 \%$. For the other cases, the solver could get the solution faster, with computation times as small as 15 minutes. The algorithmic approach is executed as a single-threaded task, given its sequential nature. The solution has been retrieved consistently in less than 600 seconds, with a clear computational advantage compared to the optimization-based approach. 
 

\begin{figure*}[t]
	\centering
	\subfloat[$\textrm{Partition}$ $\mathcal{P}_1$]{
		\includegraphics[width=0.23\linewidth]{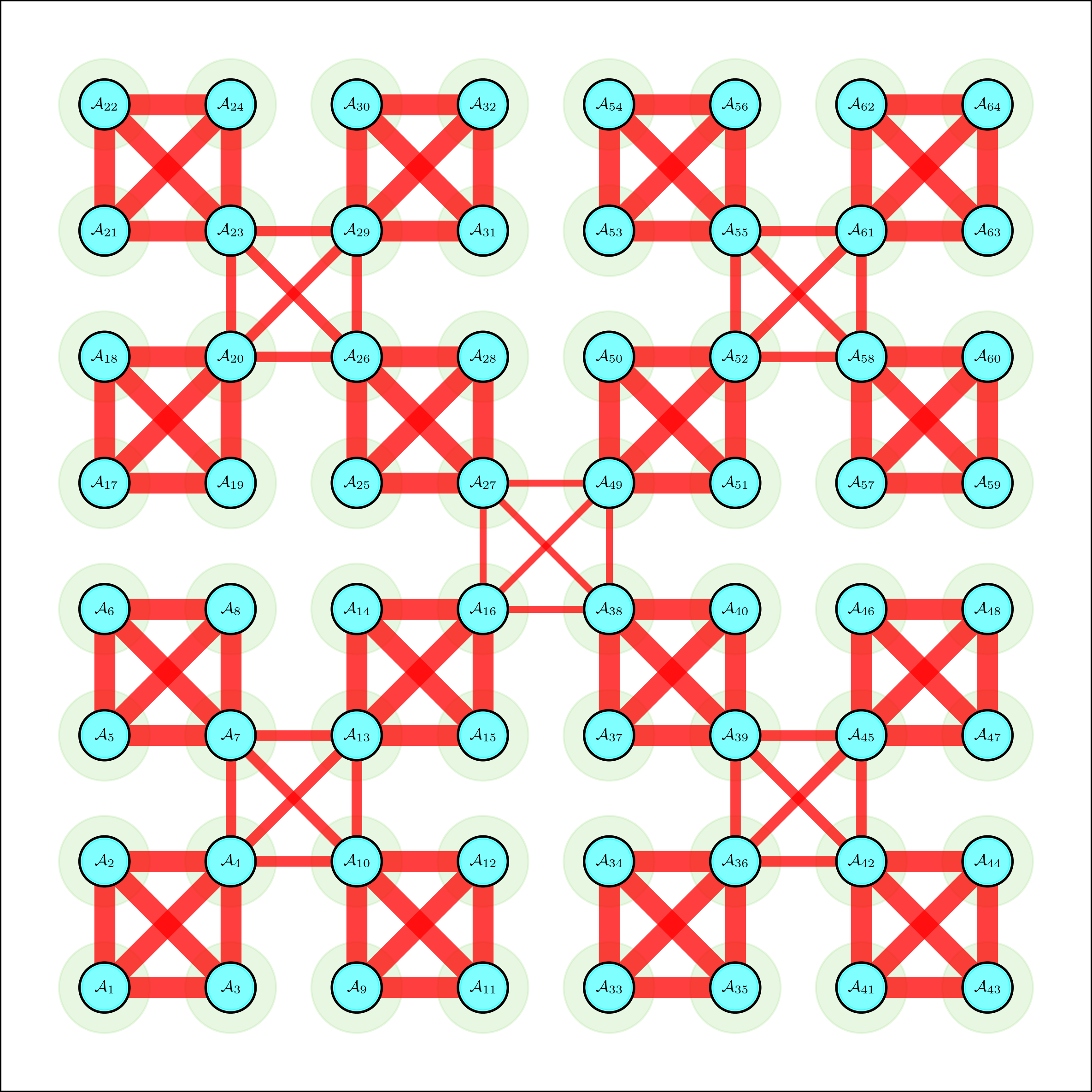}
		\label{fig:Partitioning01}
	}
	\hfill
	\subfloat[$\textrm{Partition}$ $\mathcal{P}_2$]{
		\includegraphics[width=0.23\linewidth]{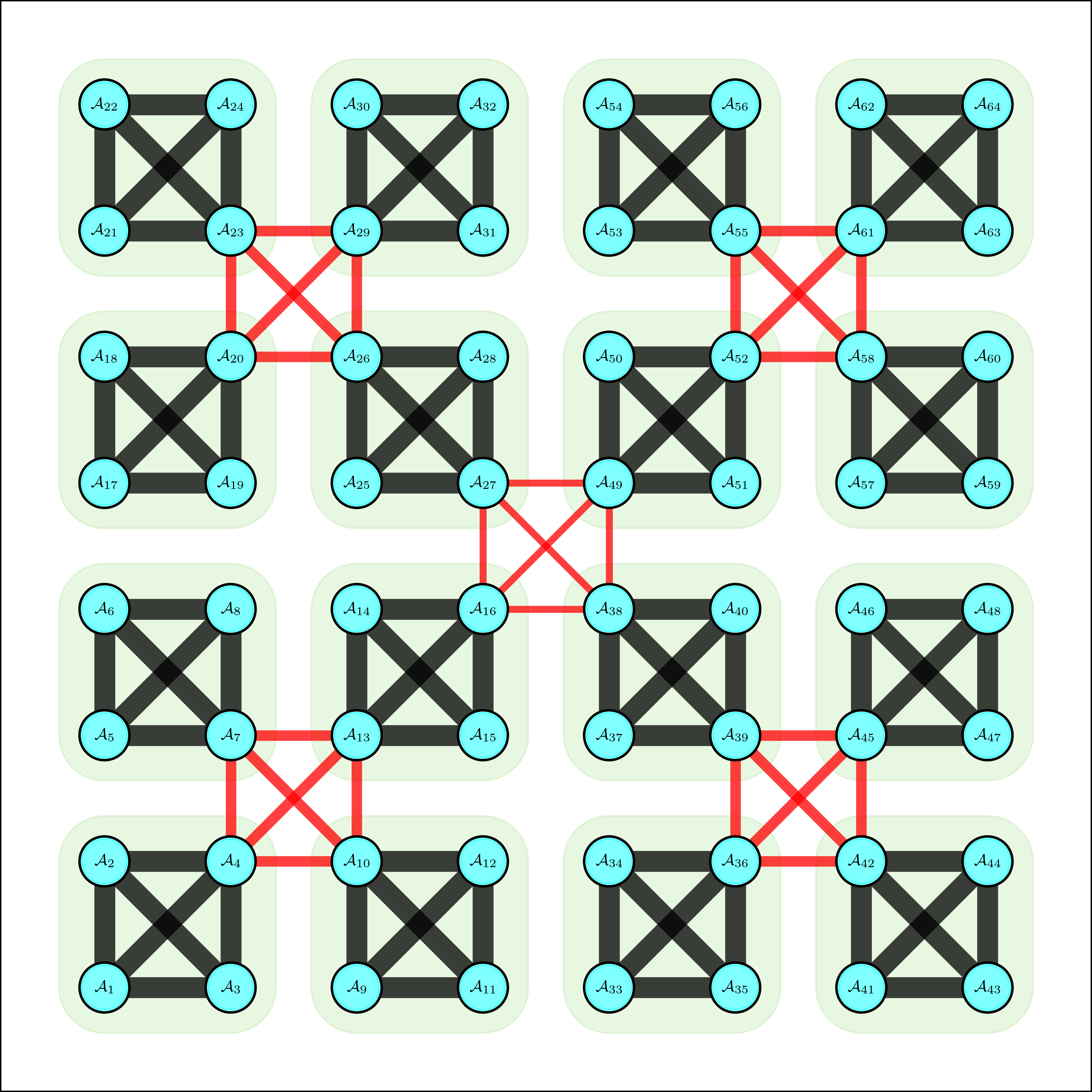}
		\label{fig:Partitioning02}
	}
	\hfill
	\subfloat[$\textrm{Partition}$ $\mathcal{P}_3$]{
		\includegraphics[width=0.23\linewidth]{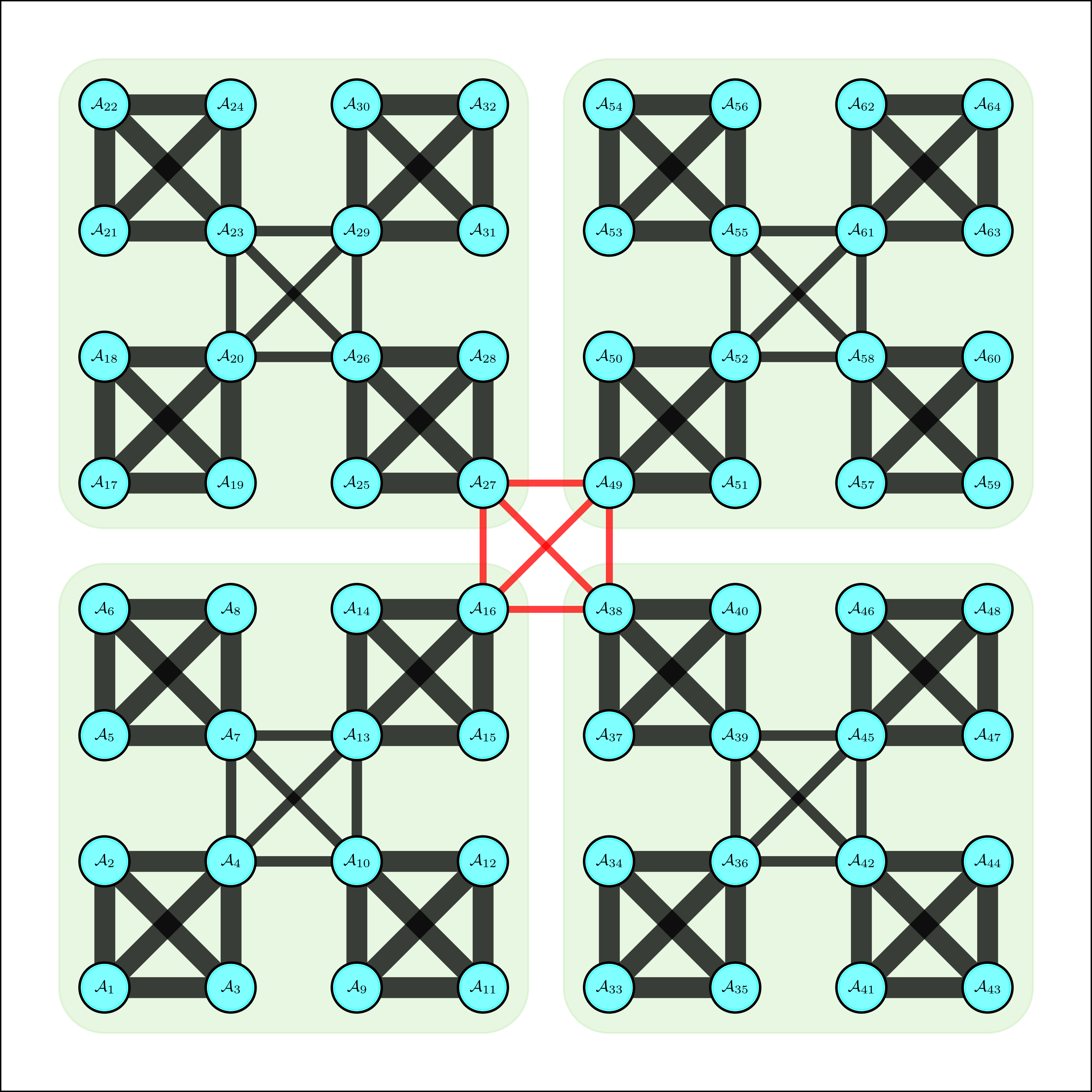}
		\label{fig:Partitioning03}
	}
	\hfill
	\subfloat[$\textrm{Partition}$ $\mathcal{P}_4$]{
		\includegraphics[width=0.23\linewidth]{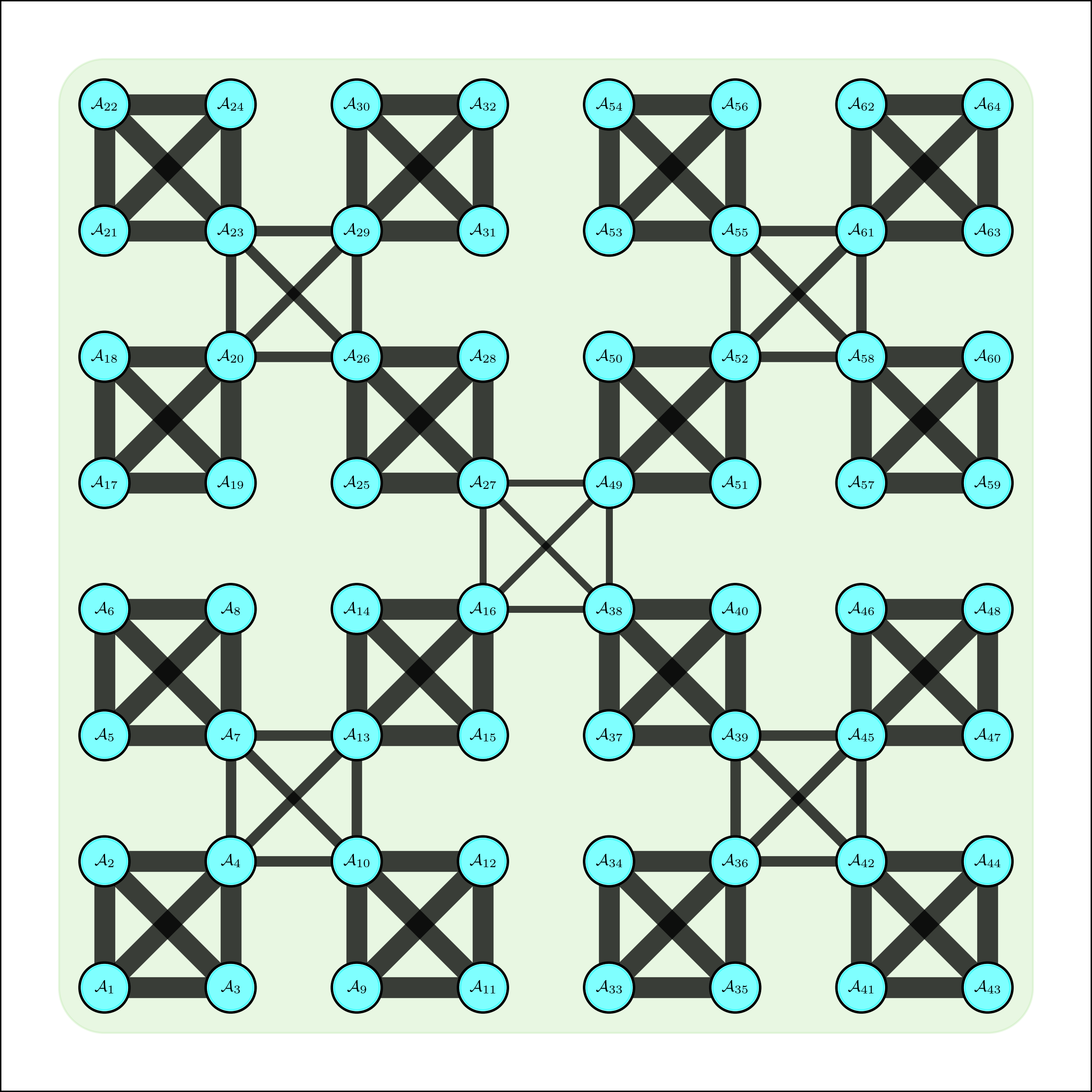}
		\label{fig:Partitioning04}
	}
	\caption{Graphs related to the partitioning of a modular network with 64 FSUs. The strength of the connection is represented by the thickness of the links. Partitions $\mathcal{P}_1$-$\mathcal{P}_4$ are obtained for values $\alpha = 10^6, 10^4, 10^2, 1$ respectively. 
	}
	\label{fig:Modular_64}
\end{figure*}

\section{Case Studies: The Role of Partitioning in Distributed Predictive Control}\label{sec:example-case-study} 
{In this section, we illustrate the applicability of the partitioning techniques developed throughout the paper to the distributed MPC (DMPC) control of networks. The DMPC architecture we will use for this aim is based on the alternating direction method of multipliers (ADMM) \cite{boyd_DistributedOptimizationStatistical_2010a}, which is a well-established framework \cite{summers_DistributedModelPredictive_2012a,rostami_ADMMbasedDistributedModel_2017b}. We will show how our partitioning technique can improve the performance of DMPC-ADMM, both in terms of stage and computation cost. This will be achieved first by applying DMPC-ADMM on the modular network of 64 FSUs presented in Sec.\ \ref{subsec:ADMM-linear} and assuming that each FSU is a linear system. 
Then in Sec.\ \ref{subsec:ADMM-hybrid}, for the same network, we assume that each FSU is a piecewise affine (PWA) system \cite{sontag_NonlinearRegulationPiecewise_1981,heemels_EquivalenceHybridDynamical_2001,tabuada_VerificationControlHybrid_2009}, for which, at each ADMM iteration, we deploy the MPC technique formulated in \cite{bemporad_ControlSystemsIntegrating_1999}. The performance of DMPC-ADMM improved w.r.t. the conventional formulation, where all FSUs are considered independently, while retaining the same pattern in computation times and costs obtained for linear systems. 
Finally, in Sec.\ \ref{sec:case-study-random}, we propose a case study of a random network of 50 FSUs for which we will compare optimization-based and algorithmic partitioning. 

\subsection{DMPC-ADMM for a modular network of linear FSUs}
\label{subsec:ADMM-linear}

\begin{figure}[t]
	\centering
	\subfloat[]{
		\includegraphics[width=.85\linewidth]{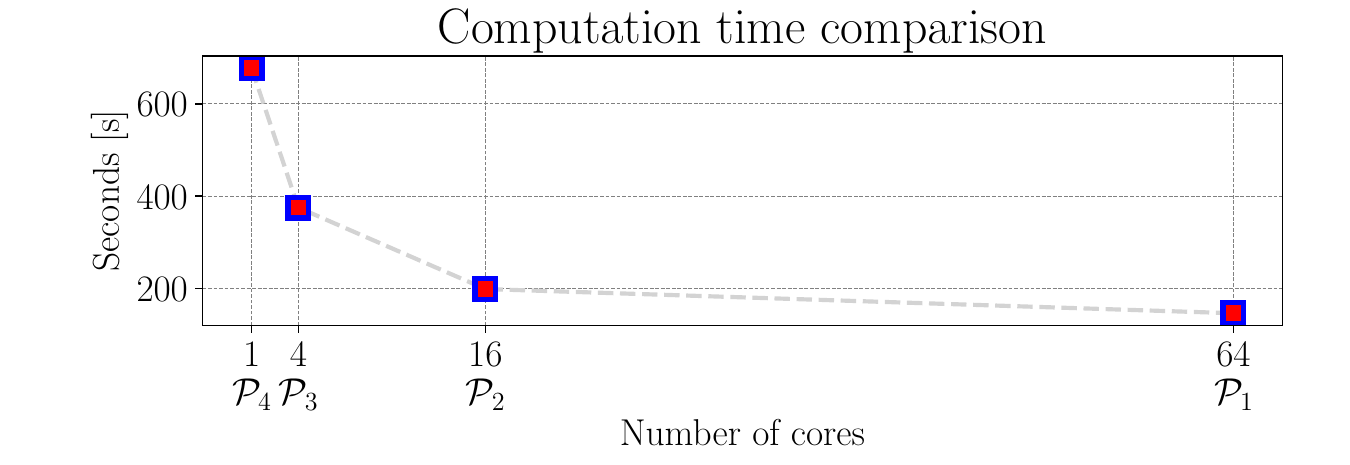}
		\label{fig:computationtimes-linear}
	}
	\hfill
	\subfloat[]{
		\includegraphics[width=.85\linewidth]{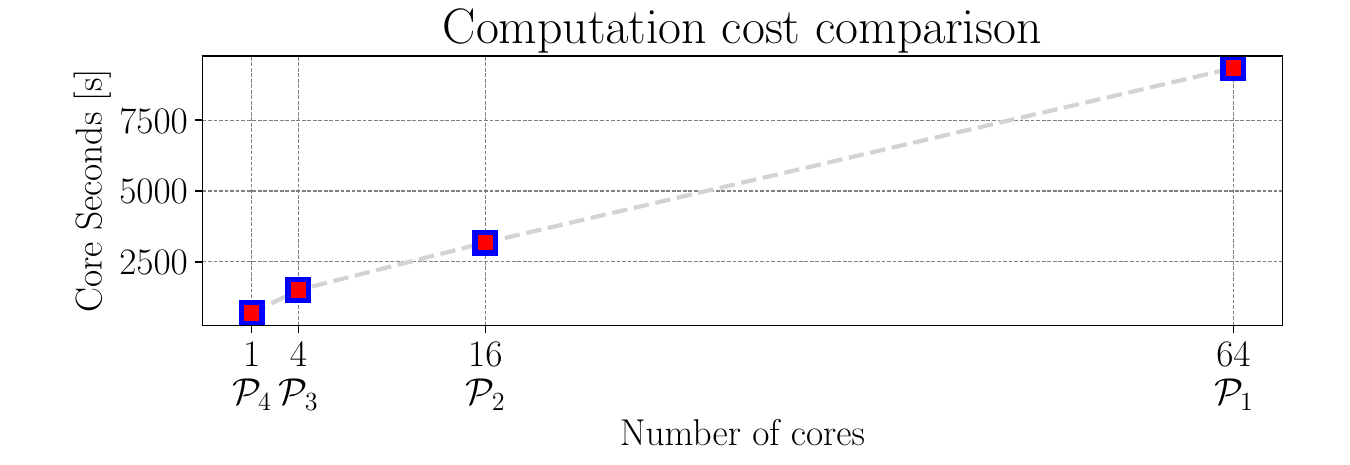}
		\label{fig:computationcost-linear}
	}
	\caption{Comparison of the computation times (Fig.\ \ref{fig:computationtimes-linear}), and computation costs (Fig.\ \ref{fig:computationcost-linear}) for the different partitions. For the former, the gray trend line highlights a diminishing return in computation speed as the number of cores increases, while for the latter the almost linear relation between the number of cores and their associated computation cost.}
\end{figure}
We consider each FSU a linear systems with matrices $A$ and $B$ such that the dynamics is stable and the network is controllable. We consider $A_{(i,i)} = 0.5$, $B_{(i,i)} = 1$, and values $A_{(i,j)} = 0.1$ or $0.01$ for strong and weak couplings respectively, as represented by the thickness of the edges in Fig.\ \ref{fig:Modular_64}. For this system, we consider a prediction horizon of $30$ in time steps. Centralized MPC \cite{scattolini_ArchitecturesDistributedHierarchical_2009, mayne_ConstrainedModelPredictive_2000a} is applied to partition $\mathcal{P}_4$, while DMPC-ADMM is used for the others, according to the formulations \cite{summers_DistributedModelPredictive_2012a,rostami_ADMMbasedDistributedModel_2017b}, and selecting parameters $\epsilon = 0.001$ (consensus threshold), $\rho = 0.01$ (penalty parameter). The control objective is to track a sinusoidal reference of unit amplitude and with frequency $w = 2 \pi/15$, with bounds $u^{[i]} \in [-0.5;0.5]$, $x^{[j]} \in [-0.9;0.9]$ $\forall i,j$. We assume a zero initial condition. Given the maturity of the MPC framework, it is not surprising that the controller achieves excellent performance for CMPC, and without any substantial variation for the DMPC-ADMM alternatives. Control objectives are met easily; thus, we omit them, but they are accessible in the repository \cite{riccardi_CodeUnderlyingPublication_2025}. Instead, we are interested in evaluating the results through the performance metrics reported in Tab.\ \ref{tab:results-linear}, where in the first column, each row corresponds to control applied using one of the partitions derived in Sec.\ \ref{subsec-partitioning_modular}. In the second column, we report the number of CSUs in the partition, corresponding to the number of computing cores necessary to deploy a fully distributed control architecture. Next, we have the value of the cumulative stage cost at the final step of the simulation. We see that the performance deteriorates minimally with a growing number of CSUs, not significantly enough to prefer the selection of a strategy w.r.t.\ another, but also indicating that each distributed approach based on the partitions obtained is effective in controlling the network. We also evaluate the different approaches through parallel computation time, and computational cost, and for the latter we use core seconds, i.e.\ the total amount of seconds the computing cores operate to deploy a fully distributed strategy. We compute core seconds by picking at each ADMM iteration the computation time of the slowest core and multiplying it by the number of cores, then summing all these values throughout the simulation. 
If we consider the parallel computation time in the fourth column, we see that the higher the number of CSUs is, the faster the DMPC-ADMM strategy is. To put the numbers in perspective, in the fifth column we consider the time required by fastest architecture $\mathcal{P}_1$ as a reference unit, and compute the ratio with the other approaches. It follows that the partition $\mathcal{P}_2$ is only $1.35$ times slower than $\mathcal{P}_4$, while the centralized MPC $4.66$ times. Therefore, even if $\mathcal{P}_1$ is much faster than CMPC, there is a diminishing return w.r.t.\ $\mathcal{P}_2$. This fact can also be observed in Fig.\ \ref{fig:computationtimes-linear}, where we report the computation times. Finally, we compare the computational costs in the last two columns of Tab.\ \ref{tab:results-linear}. 
Accordingly, considering CMPC in $\mathcal{P}_1$ the least expensive strategy, the computation cost increase trend is almost linear in the number of cores.
This fact can be seen in Fig.\ \ref{fig:computationcost-linear}, where we plot the computation costs. 
The analysis of the results shows that while the computation cost increases almost linearly with the number of CSUs, the computation time decreases exponentially, highlighting a pattern of diminishing returns. Accordingly, it is comparatively more expensive to increase the number of CSUs from the computation perspective.
In this case, partition $\mathcal{P}_2$ is $1.35$ times slower than $\mathcal{P}_1$ while being $2.94$ cheaper, and requiring 48 computing cores less to be executed. Considerations of this nature can vastly differentiate the selection of different approaches in practice. 
While DMPC-ADMM performance as a function of a selected partition is not an issue for linear systems, computation times can be. 
As an example, we considered a partition manually selected to maximize the inter-CSU coupling, accessible at \cite{riccardi_CodeUnderlyingPublication_2025}. 
As a result of our testing of the same control problem of the other cases, even if the ADMM approach can converge, we get an estimated parallel computation time of $160$ hours for the same simulation compared to less than 400 seconds for the other approaches, a clear demonstration of the value of proper selection of partitions.
%

\begin{table*}
	\centering
		\caption{ Comparison of DMPC-ADMM performance applied to a network of linear systems for different partitioning strategies}
	\begin{tabular}{c | c | c |c | c | c | c }
		\hline 
		\textbf{Partition}  &   \textbf{Cores}  & \textbf{Cost function value\ } & \textbf{Computation time [$s$]}  & \textbf{Computation time  ratio} &  \textbf{Core seconds [$s$]} & \textbf{Core seconds ratio}\\
		\hline 
		$\mathcal{P}_{4}$& 1 & 465.4703 &680.88 & 4.6619 & 680.88 & 1.0000 \\ 
		$\mathcal{P}_{3}$& 4 & 465.4709 &369.56 & 2.5303 & 1478.23& 2.1710 \\ 
		$\mathcal{P}_{2}$& 16 & 465.5153 &198.52 & 1.3592 & 3176.39& 4.6651 \\ 
		$\mathcal{P}_{1}$&  64& 465.5158 &146.05 & 1.0000 & 9347.37& 13.7283 \\ 
		\hline
	\end{tabular}
	\label{tab:results-linear}
\end{table*}
\begin{table*}
	\centering
	\caption{ Comparison of DMPC-ADMM performance applied to a network of hybrid systems for different partitioning strategies}
	\begin{tabular}{c | c | c |c | c | c | c }
		\hline 
		\textbf{Partition}  &   \textbf{Cores}  & \textbf{Cost function value\ } & \textbf{Computation time [$s$]}  & \textbf{Computation time  ratio} &  \textbf{Core seconds [$s$]} & \textbf{Core seconds ratio}\\
		\hline 
		$\mathcal{P}_{4}$& 1 & 24949.16 & 849.57 & 2.7691 & 849.57 & 1.0000 \\ 
		$\mathcal{P}_{3}$& 4 & 24945.76 & 632.35 & 2.0611 & 2529.41& 2.9772 \\ 
		$\mathcal{P}_{2}$& 16 & 24945.77 & 488.05 & 1.5907 & 7808.80 & 9.1914 \\ 
		$\mathcal{P}_{1}$&  64& 24946.04 & 306.80 & 1.0000 & 19635.58 & 23.1123 \\ 
		\hline
	\end{tabular}
	\label{tab:results-hybrid}
\end{table*}

\subsection{DMPC-ADMM for a modular network of hybrid FSUs}
\label{subsec:ADMM-hybrid}
{
We now assume that each CSU in the network Fig.\ \ref{fig:Modular_64} is a hybrid system, specifically a PWA system of the form:
\begin{align} \label{eq:hybrid}
	& x^{[i]}(k+1) = 0.5 x^{[i]}(k) + u^{[i]}(k) + \nonumber \\ & + \sum_{j\in \mathcal{N}_i}  A_{(i,j)}x^{[j]}(k) \qquad \text{if} \quad  x^{[i]}(k) \geq 0 \\
	& x^{[i]}(k+1) = -0.5 x^{[i]}(k) + u^{[i]}(k) + \nonumber  \\ & + \sum_{j\in \mathcal{N}_i}  A_{(i,j)}x^{[j]}(k) \qquad \text{if} \quad  x^{[i]}(k) < 0 
\end{align}
{ where, as in the previous case, $A_{(i,j)} = 0.1$ or $0.01$ as represented in Fig.\ \ref{fig:Modular_64}. Also, we retain the bounds $u^{[i]} \in [-0.5;0.5]$, $x^{[j]} \in [-0.9;0.9]$ $\forall i,j$. The objective is to track a sinusoidal reference with unit amplitude and frequency $w = 2\pi /20$ from a zero initial condition. Here, we reduce the prediction horizon to 6 time steps to reduce the complexity of the problems.
In this case we first convert the system into Mixed Logical Dynamical (MLD) form \cite{heemels_EquivalenceHybridDynamical_2001}, and then we apply MPC as described in \cite{bemporad_ControlSystemsIntegrating_1999} for centralized control, and use the same procedure in combination with the DMPC-ADMM strategy \cite{boyd_DistributedOptimizationStatistical_2010a,summers_DistributedModelPredictive_2012a,rostami_ADMMbasedDistributedModel_2017b}.
In this case, at each step we are required to solve a Mixed Integer Quadratic Program (MIQP) \cite{conforti_IntegerProgramming_2014} to obtain the control action. Even if for nonconvex and possibly nonsmooth optimization it is not always guaranteed that the ADMM algorithm will converge \cite{wang_GlobalConvergenceADMM_2019}, similar approaches have been deployed successfully in literature \cite{luan_DecompositionDistributedOptimization_2020a}. We use the same hardware specified in Sec.\ \ref{subsec:ADMM-linear}. We report the results of the simulations in Figs.\ \ref{fig:inputPWA}, \ref{fig:statePWA} and in Tab.\ \ref{tab:results-hybrid}. The results show a marginally improved cumulative stage cost. Also, the patterns found in the previous} section for linear systems regarding computation times and costs repeat here. 
}


\begin{figure}[t]
	\centering
	\subfloat[]{
		\includegraphics[width=.85\linewidth]{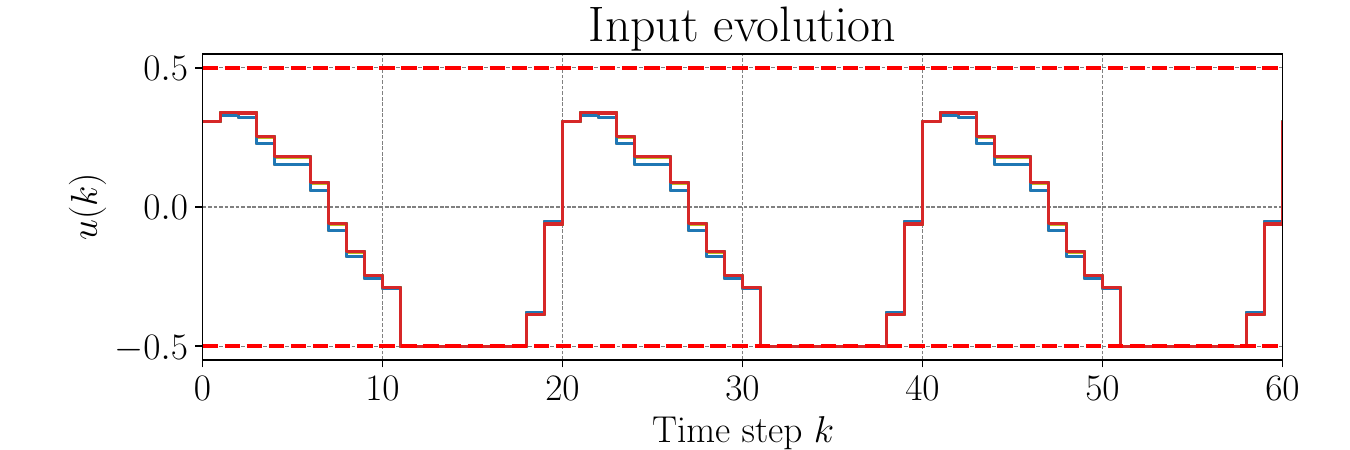}
		\label{fig:inputPWA}
	}
	\hfill
	\subfloat[]{
		\includegraphics[width=.85\linewidth]{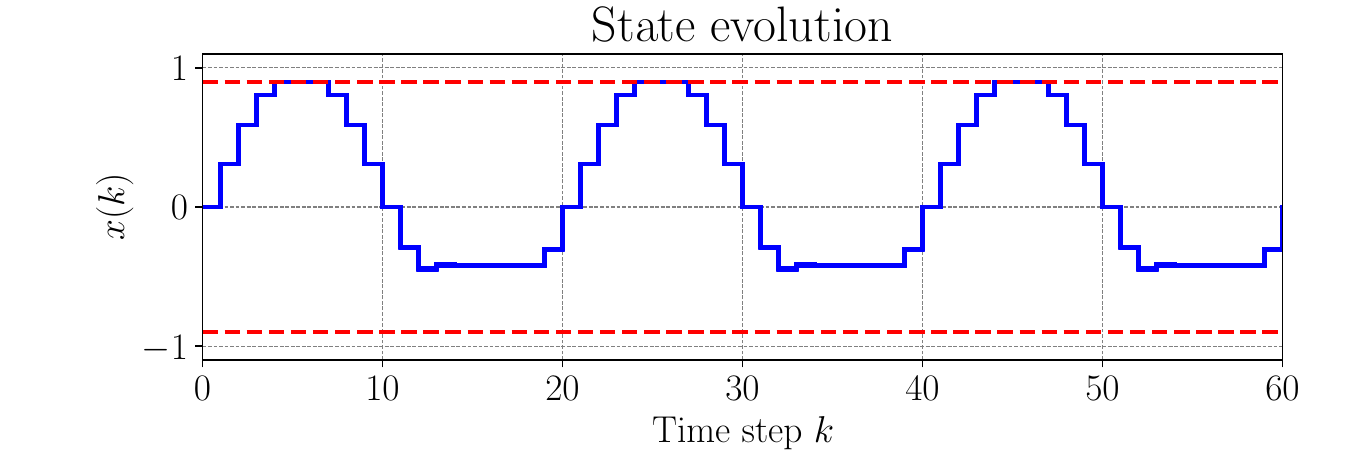}
		\label{fig:statePWA}
	}
	\caption{Evolution of the inputs (Fig.\ \ref{fig:inputPWA}) and of the states (Fig.\ \ref{fig:statePWA}) for the network of PWA FSUs. For the former, all signals are approximately the same and they respect the constraints.For the latter, when the value of a state crosses the threshold at zero, we have a change of dynamics, which is reflected in a different behavior. Accordingly, the input signal changes too, but reaches saturation before the state is able to reproduce the reference. 
	}
\end{figure}

%

\subsection{DMPC-ADMM for a random network of 50 hybrid systems}\label{sec:case-study-random}

{\CB We consider a network with 50 FSUs with hybrid dynamics \eqref{eq:hybrid} that are connected through a randomly generated topology, reproducible through the software in \cite{riccardi_CodeUnderlyingPublication_2025}. The network is in Fig.\ \ref{fig:Random-50-topology}. We apply both optimization-based and algorithmic partitioning, obtaining different configurations of CSUs for which we will test CMPC and DMPC-ADMM control strategies as above. All the partitions obtained are reported in Fig.\ \ref{fig:Partition_01_Random_Alg_Opt} - \ref{fig:Partition_02_Random_Alg_Opt}. For both partitioning approaches, the endpoints of the range of $\alpha$ return two partitions that we denote by $\mathcal{P}^{\text{ADMM}}$ (with 50 individual CSUs) and $\mathcal{P}^{\text{CMPC}}$ (with one CSU for the whole network). The superscript notation refers to the fact that we deploy the conventional DMPC-ADMM strategy for the former partition, and for the latter CMPC. 
Selecting the weights as described in Sec.\ \ref{subsec-partitioning_modular}, for optimization-based partitioning we obtain five partitions $\mathcal{P}^{\text{Opt}}_i$, $i = 0,\ldots,4$, for $\alpha$ varying, which are in Figs.\ \ref{fig:Partition_01_Random_Alg_Opt}, \ref{fig:Partition_01_Random_Opt}, \ref{fig:Partition_02_Random_Opt}, \ref{fig:Partition_03_Random_Opt}, \ref{fig:Partition_02_Random_Alg_Opt}. For algorithmic partitioning we have  the six partitions $\mathcal{P}^{\text{Alg}}_i$, $i = 0,\ldots,5$, shown in Figs.\ \ref{fig:Partition_01_Random_Alg_Opt}, \ref{fig:Partition_02_Random_Alg}, \ref{fig:Partition_03_Random_Alg}, \ref{fig:Partition_04_Random_Alg}, \ref{fig:Partition_05_Random_Alg}, \ref{fig:Partition_02_Random_Alg_Opt}. Respective weights are indicated below the figures, where the CSUs are represented by the colored areas collecting together FSUs.
Given such partitions, we apply the same DMPC strategy deployed in Sec.\ \ref{subsec:ADMM-hybrid}. This time, we assume a different initial condition for each CSU and a different sinusoidal reference to track, both randomly generated. Data resulting from the simulations and scripts to reproduce them can be accessed at \cite{riccardi_CodeUnderlyingPublication_2025}. Reporting the graphs of each simulation is prohibitive. Therefore, here we propose in Figs.\ \ref{fig:Error}, \ref{fig:MaxState} respectively the evolution of the maximum error and the maximum absolute value of the states across all FSUs for each partition, showing that the controller can achieve the objectives. Results for all partitions are in Tab.\ \ref{tab:results-random}. Here we notice that, while for $\mathcal{P}^{\text{CMPC}}$ we still get the best performance at the expense of the highest computation time, for $\mathcal{P}^{\text{ADMM}}$ we have a different situation w.r.t.\ what we found previously in Sec.\ \ref{subsec:ADMM-hybrid}. For $\mathcal{P}^{\text{ADMM}}$, there is now a substantial difference in performance w.r.t.\ the other partitions, which poses the question of the impact of topology in network control.  
The best-performing strategies are now in order $\mathcal{P}^{\text{Alg}}_3$ and $\mathcal{P}^{\text{Opt}}_2$, with a difference in the cost value of $0.16\%$, $0.24\%$ respectively, but a significantly lower computation time for $\mathcal{P}^{\text{Opt}}_2$, which can be considered as the most suitable partition in this case. To further compare the results, Figs.\ \ref{fig:computationtimes-random}, \ref{fig:computationcost-random} show the graphs relating computation time and cost to the number of CSUs in the network, and consequently of computing cores necessary for the implementation of the distributed architecture, as done in Sec.\ \ref{subsec:ADMM-hybrid}. Here, we see that for optimization-based partitions, we still get the same trend in the scaling of the computation w.r.t.\ the number of CSUs as in Fig.\ \ref{fig:computationtimes-linear}, and an almost linear scaling in the computation cost as in Fig.\ \ref{fig:computationcost-linear}. However, for algorithmic partitioning, we get two outliners with $\mathcal{P}^{\text{AOpt}}_4$ and $\mathcal{P}^{\text{AOpt}}_5$, both relatively high in the computation time and cost, but different in the overall performance. This analysis suggests that partitions obtained with the optimization-based approach tend to provide more reliable performances according to these evaluation criteria. However, algorithmic partition remains a valid alternative, often providing good results with low computational cost.

}

\begin{figure}[t]
	\centering
	\subfloat[]{
		\includegraphics[width=.85\linewidth]{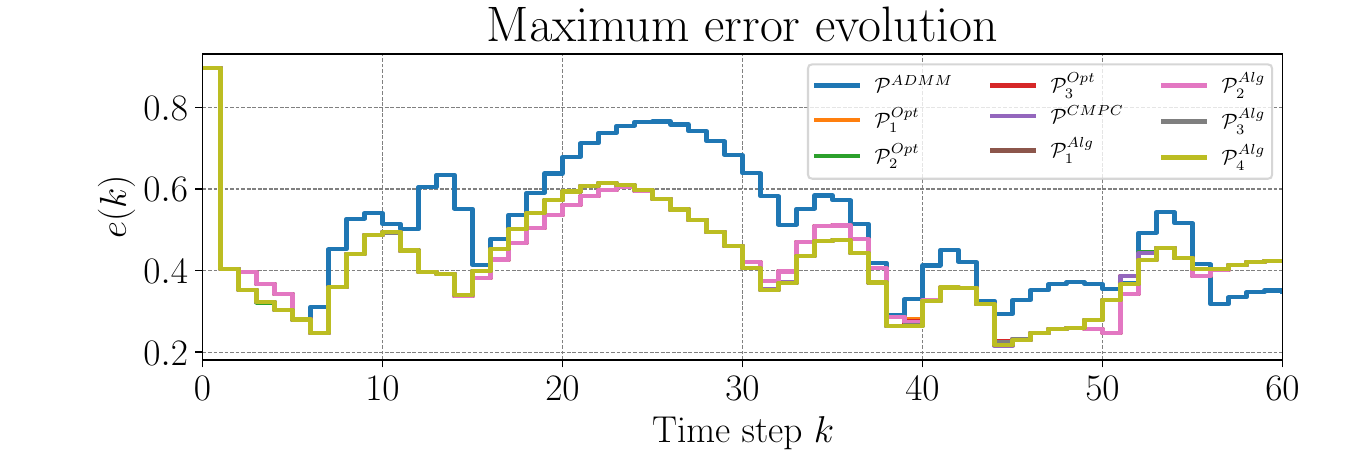}
		\label{fig:Error}
	}
	\hfill
	\subfloat[]{
		\includegraphics[width=.85\linewidth]{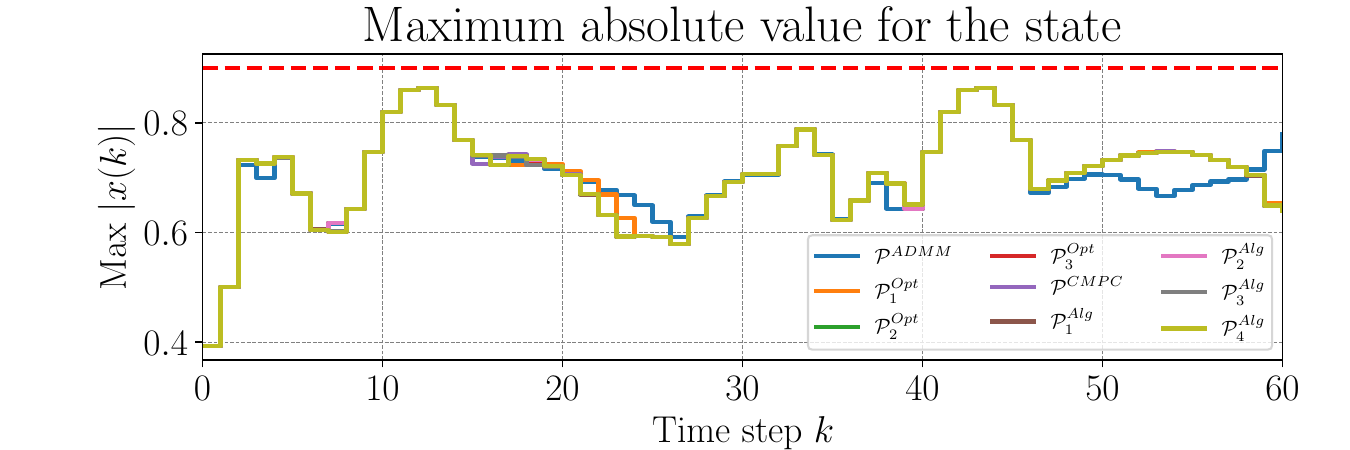}
		\label{fig:MaxState}
	}
	\caption{Evolution of the maximum error (Fig.\ \ref{fig:Error}) and absolute value (Fig.\ \ref{fig:MaxState}) across all CSUs for different partitions.
	}
\end{figure}

%

\begin{figure*}[t]
	\centering
	\subfloat[$\mathcal{P}^{\textrm{ADMM}}\equiv\mathcal{P}^{\textrm{Opt}}_0\equiv\mathcal{P}^{\textrm{Alg}}_0$]{
		\includegraphics[width=0.24\linewidth]{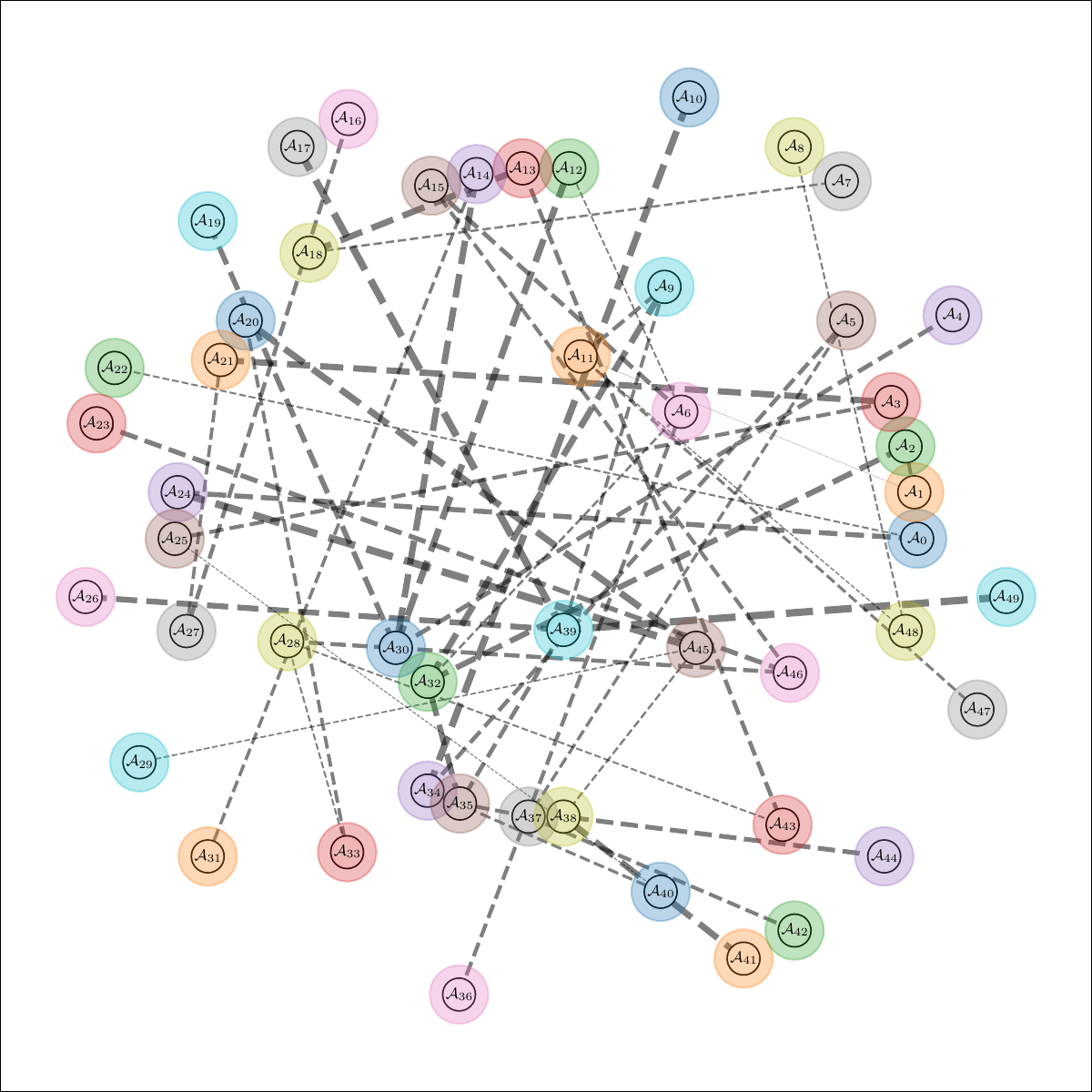}
		\label{fig:Partition_01_Random_Alg_Opt}
	}
	\hfill
	\subfloat[$\mathcal{P}^{\textrm{Opt}}_1$]{
		\includegraphics[width=0.24\linewidth]{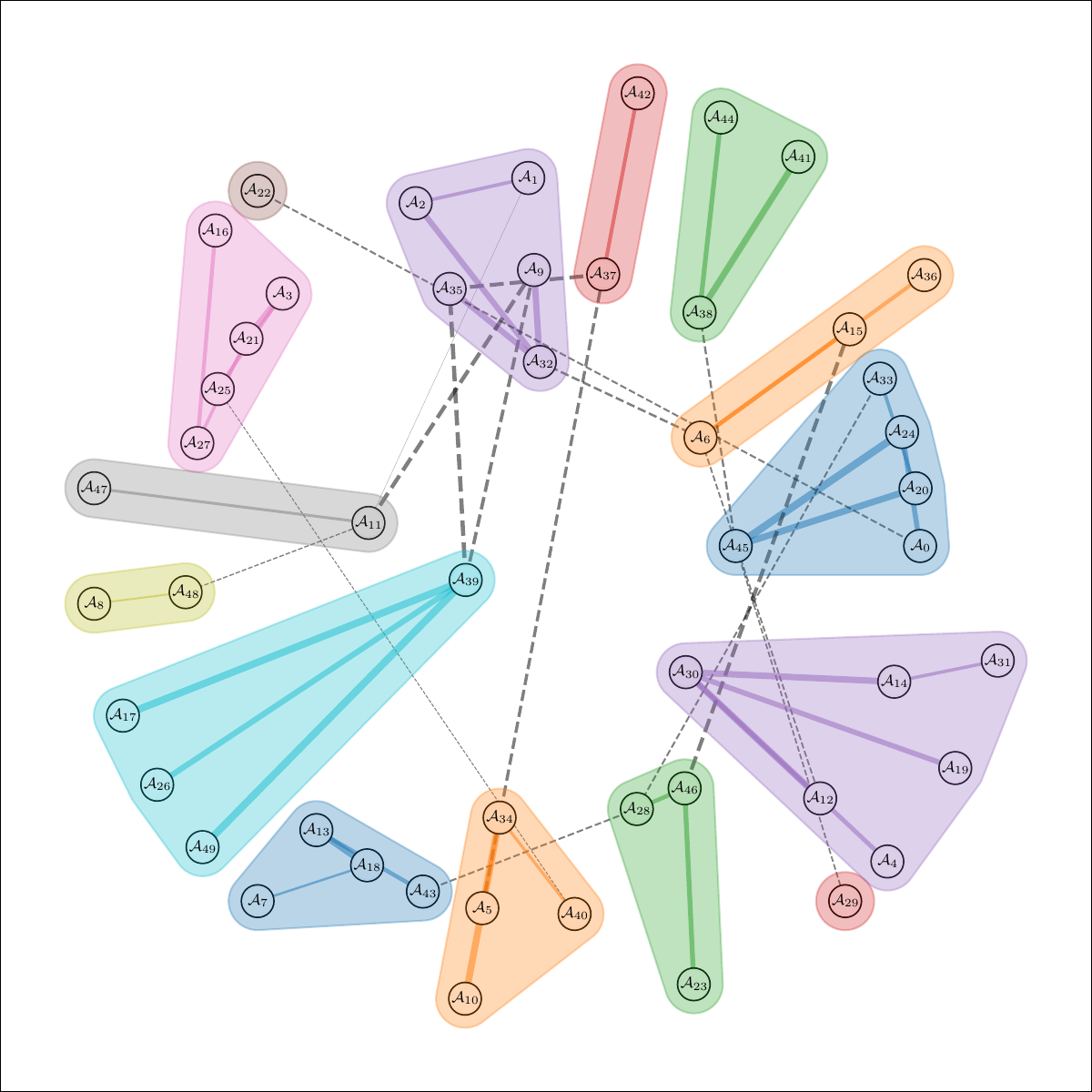}
		\label{fig:Partition_01_Random_Opt}
	}
	\hfill
	\subfloat[$\mathcal{P}^{\textrm{Alg}}_1$]{
		\includegraphics[width=0.24\linewidth]{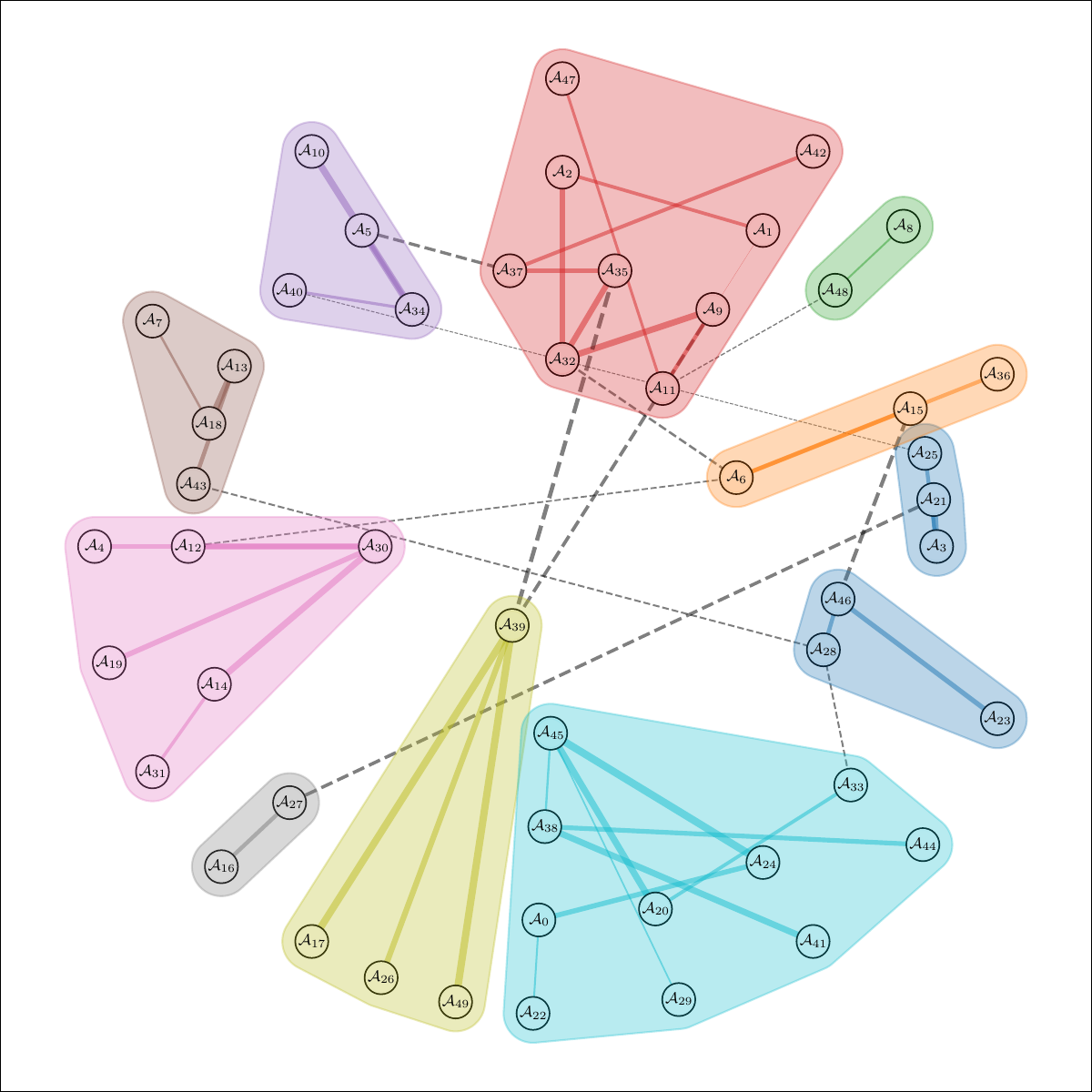}
		\label{fig:Partition_02_Random_Alg}
	}
	\hfill
	\subfloat[$\mathcal{P}^{\textrm{Alg}}_2$]{
		\includegraphics[width=0.24\linewidth]{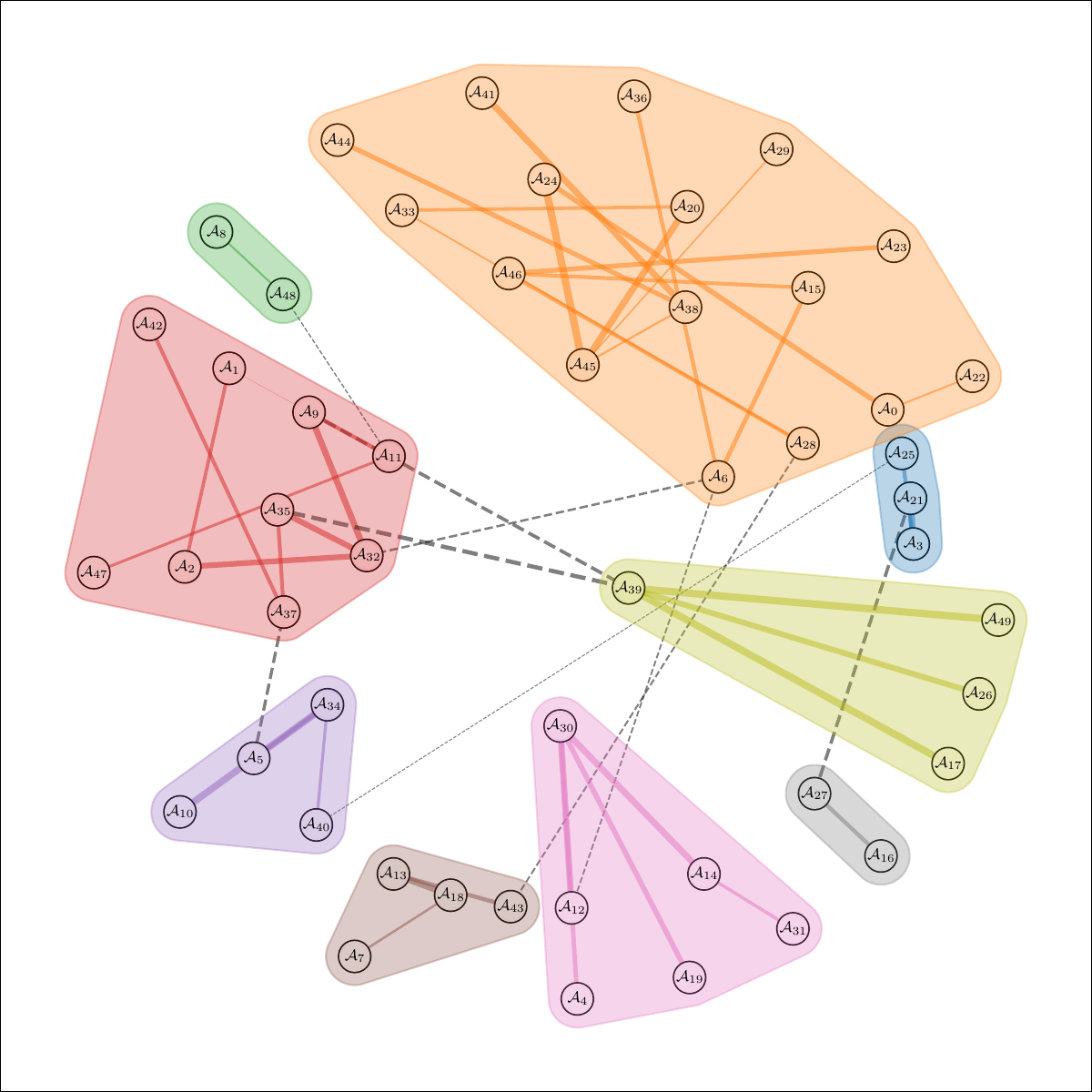}
		\label{fig:Partition_03_Random_Alg}
	}
	\hfill
	\subfloat[$\mathcal{P}^{\textrm{Opt}}_2$]{
		\includegraphics[width=0.24\linewidth]{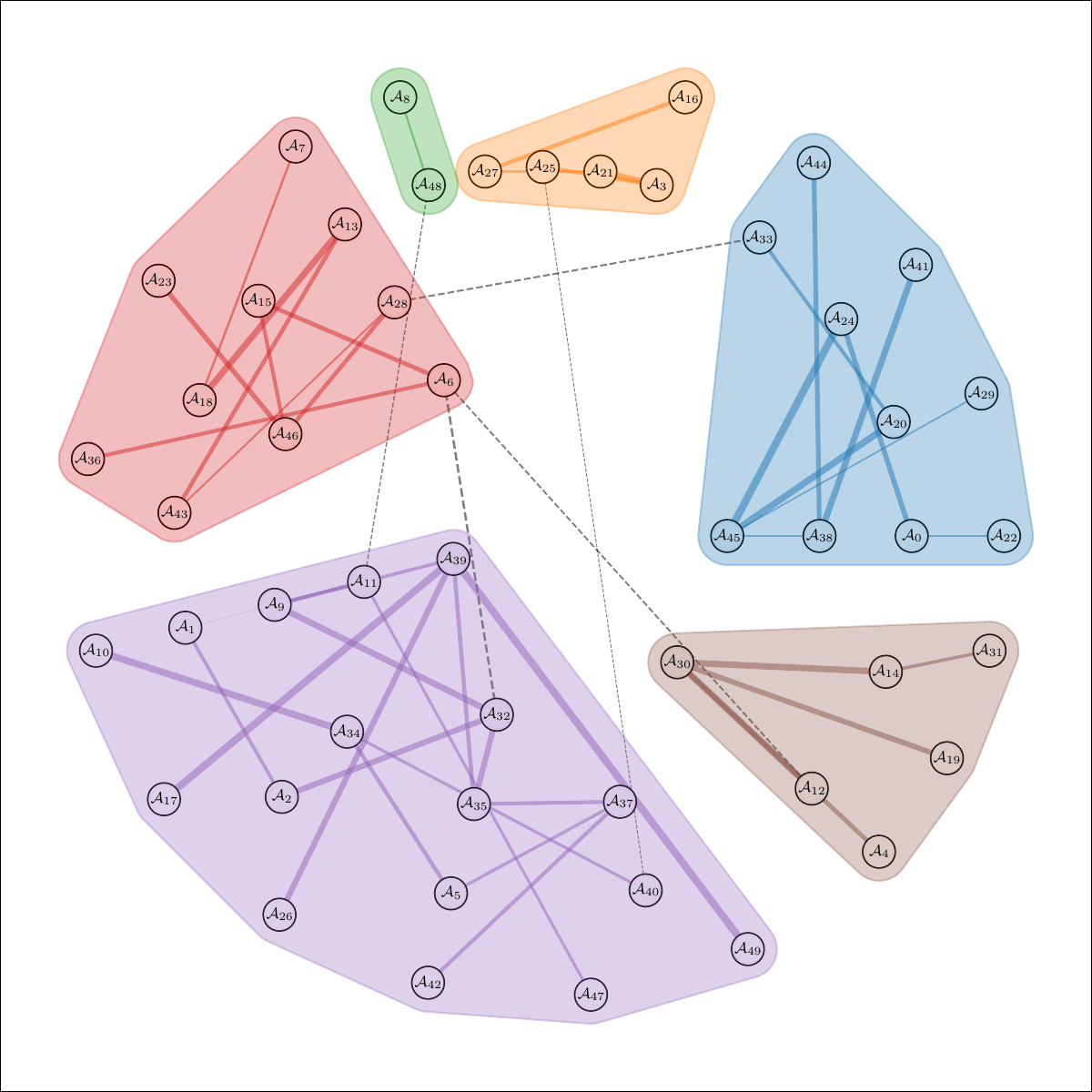}
		\label{fig:Partition_02_Random_Opt}
	}
	\hfill
	\subfloat[$\mathcal{P}^{\textrm{Alg}}_3$]{
		\includegraphics[width=0.24\linewidth]{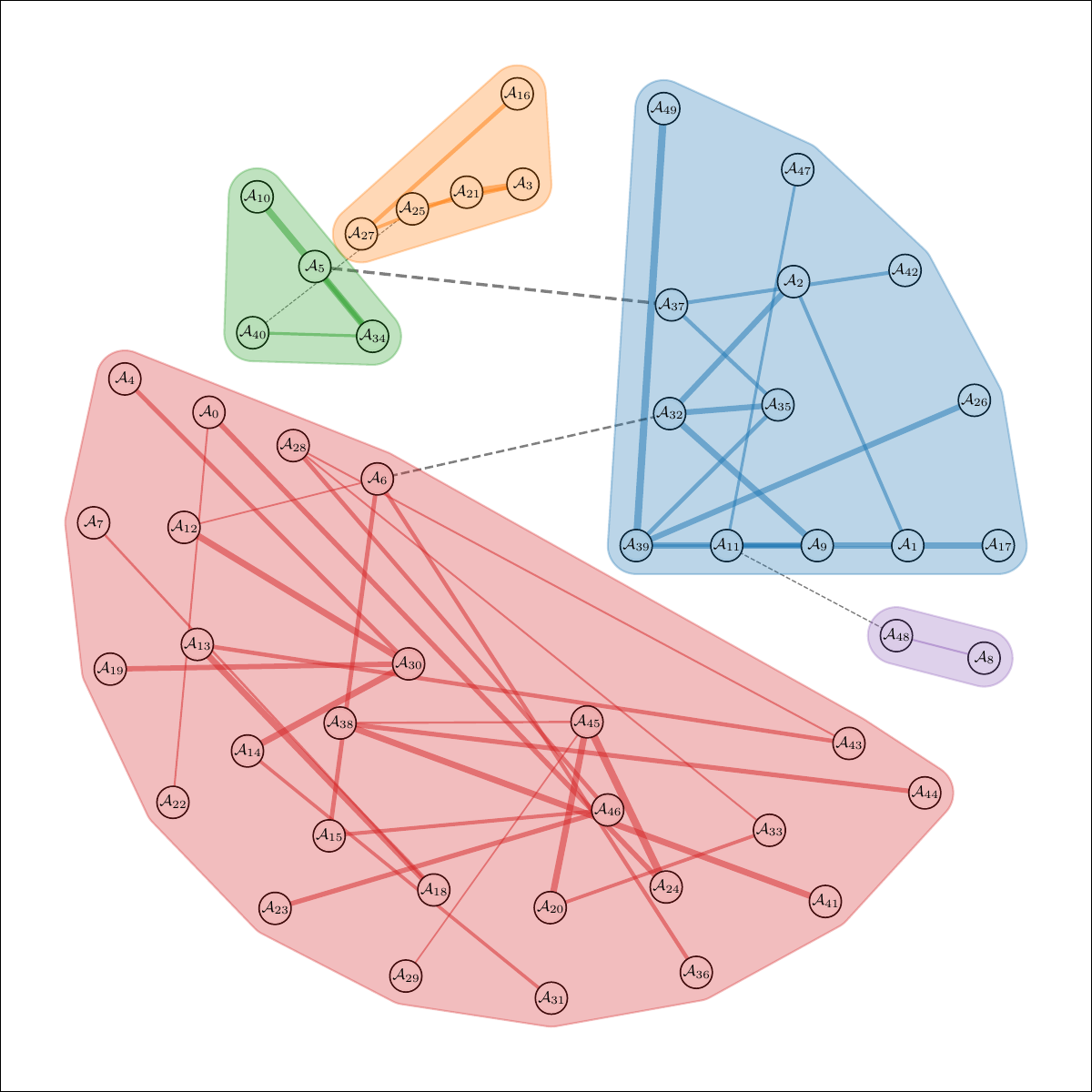}
		\label{fig:Partition_04_Random_Alg}
	}
	\hfill
	\subfloat[$\mathcal{P}^{\textrm{Alg}}_4$]{
		\includegraphics[width=0.24\linewidth]{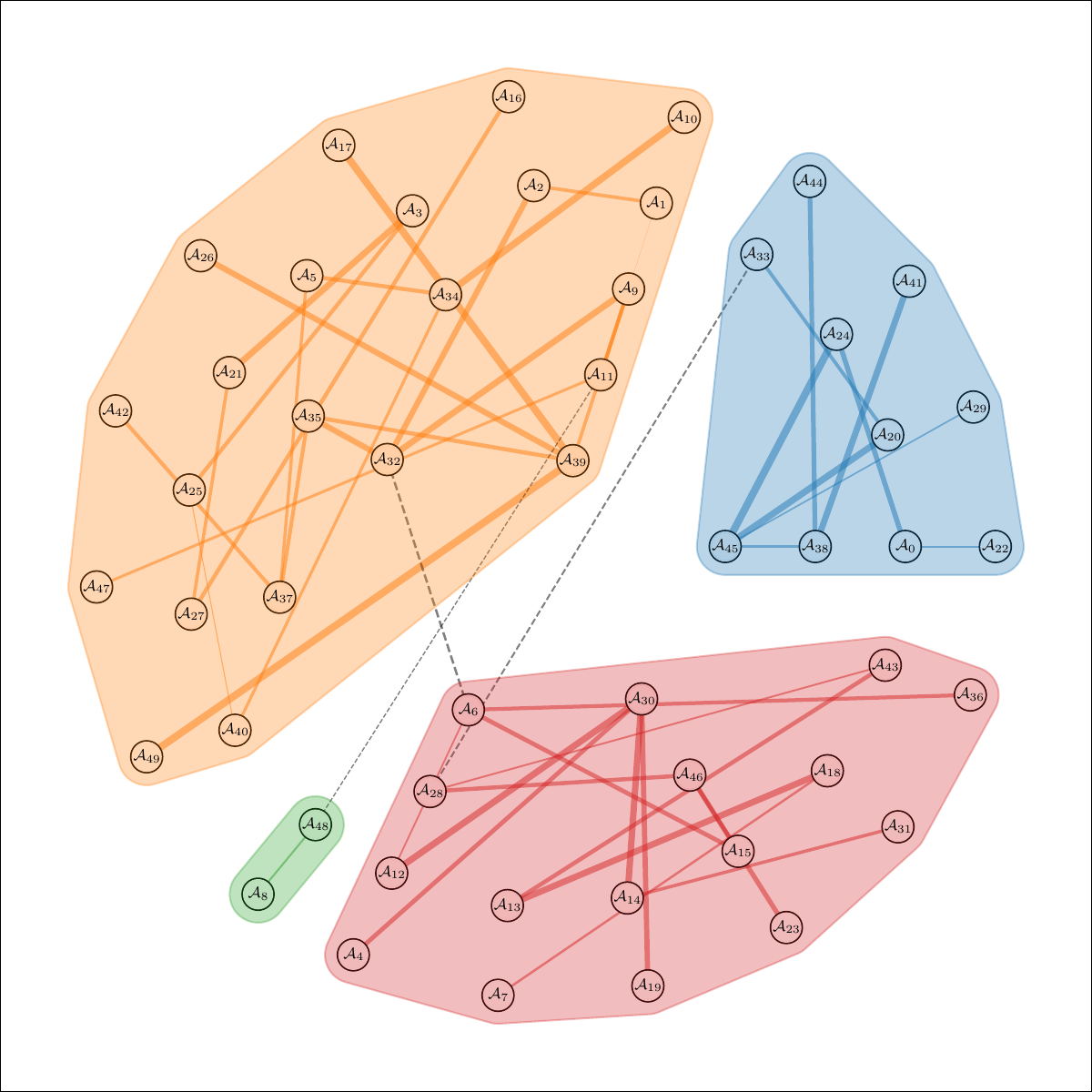}
		\label{fig:Partition_05_Random_Alg}
	}
	\hfill
	\subfloat[$\mathcal{P}^{\textrm{Opt}}_3$]{
		\includegraphics[width=0.24\linewidth]{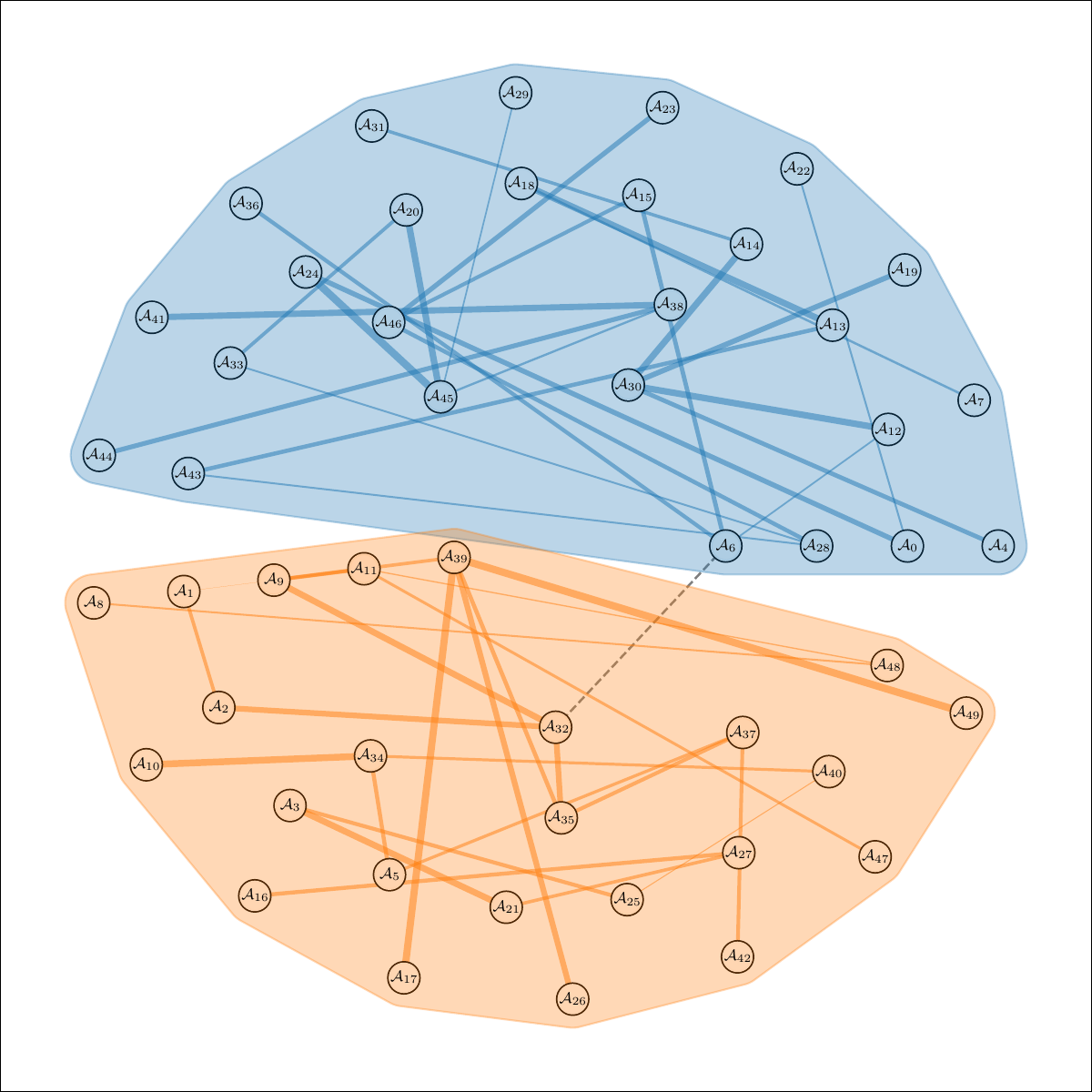}
		\label{fig:Partition_03_Random_Opt}
	}
	\hfill
	\subfloat[$\mathcal{P}^{\textrm{CMPC}}\equiv\mathcal{P}^{\textrm{Opt}}_4\equiv\mathcal{P}^{\textrm{Alg}}_5$]{
		\includegraphics[width=0.24\linewidth]{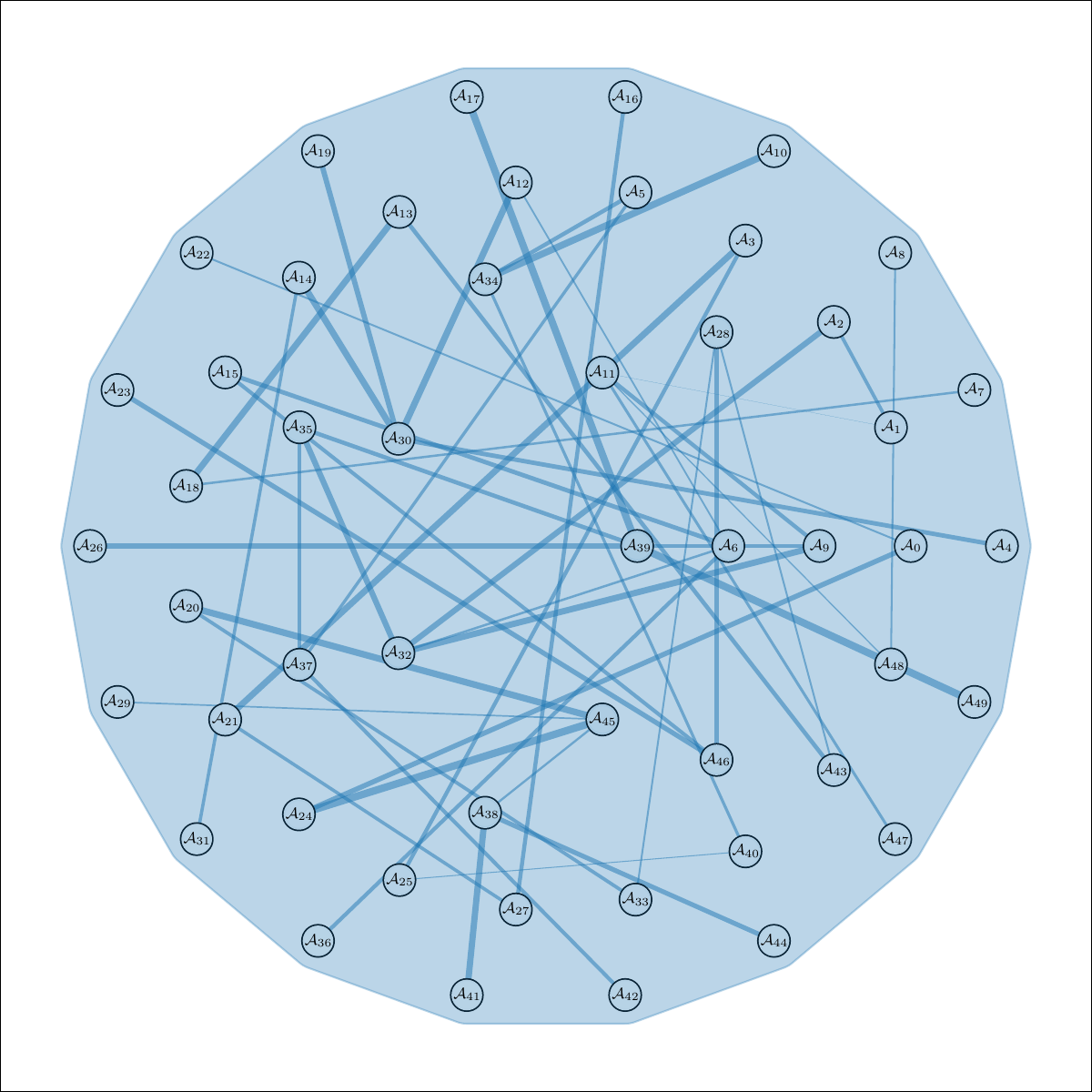}
		\label{fig:Partition_02_Random_Alg_Opt}
	}
	\caption{Partitions of a random network with 50 FSUs. In the following, for each partition it is indicated the number of resulting CSUs, and the granularity parameter $\alpha$ used in optimization-based and algorithmic partitioning. $\mathcal{P}^{\textrm{ADMM}}\equiv\mathcal{P}^{\textrm{Opt}}_0\equiv\mathcal{P}^{\textrm{Alg}}_0$: 50 CSUs, $\alpha^{\text{Opt}}_0 = 4.3\cdot 10^6$, $\alpha^{\textrm{Alg}}_0 =  100$; $\mathcal{P}^{\textrm{Opt}}_1$: 15 CSUs, $\alpha^{\textrm{Opt}}_1 = 4.3\cdot 10^4$; $\mathcal{P}^{\textrm{Alg}}_1$: 11 CSUs, $\alpha^{\textrm{Alg}}_1 = 30$; $\mathcal{P}^{\textrm{Alg}}_2$: 9 CSUs, $\alpha^{\textrm{Alg}}_2 = 20$; $\mathcal{P}^{\textrm{Opt}}_2$: 6 CSUs, $\alpha^{\textrm{Opt}}_2 = 4.3\cdot 10^3$; $\mathcal{P}^{\textrm{Alg}}_3$: 5 CSUs, $\alpha^{\textrm{Alg}}_3 = 10$; $\mathcal{P}^{\textrm{Alg}}_4$: 4 CSUs, $\alpha^{\textrm{Alg}}_4 = 1$; $\mathcal{P}^{\textrm{Opt}}_3$, 6 CSUs, $\alpha^{\textrm{Opt}}_3 = 4.3\cdot 10^2$; $\mathcal{P}^{\textrm{CMPC}}\equiv\mathcal{P}^{\textrm{Opt}}_4\equiv\mathcal{P}^{\textrm{Alg}}_5$, 1 CSU, $\alpha^{\textrm{Opt}}_4 = 4.3\cdot 10^1$, $\alpha^{\textrm{Alg}}_5 = 10^{-13}$.
	}
    \label{fig:Random-50-topology}
\end{figure*}

\begin{table*}
	\centering
	\caption{ Comparison of DMPC-ADMM performance applied to a random network of hybrid systems for different partitioning strategies}
	\begin{tabular}{c | c |c| c |c | c | c | c }
		\hline 
		\textbf{Partition}  &   \textbf{Cores}  & \textbf{Cost fun.\ value\ } & \textbf{Opt.\ loss $\%$} & \textbf{Comp.\ time [$s$]}  & \textbf{Comp.\ time  ratio} &  \textbf{Core seconds [$s$]} & \textbf{Core seconds ratio}\\
		\hline 
		$\mathcal{P}^{\text{CMPC}}$& 1 & 6899.9750 &0.00 &2628.04 & 26.4870 & 2628.04 & 1.3736 \\ 
		$\mathcal{P}^{\text{ADMM}}$& 50 & 7749.2102 & 12.31 &99.22 & 1.0000 & 4960.99& 2.5930 \\
		$\mathcal{P}^{\text{Opt}}_1$& 15 & 6976.9122 & 1.12 &279.25 & 2.8145 & 4188.73 & 2.1893 \\
		$\mathcal{P}^{\text{Opt}}_2$&  6& 6916.7114 & 0.24 &436.29 &  4.3972 & 2617.71 & 1.3682 \\ 
		$\mathcal{P}^{\text{Opt}}_3$&  2& 6918.2608 & 0.27 &1368.36 &  13.7918 & 2736.72 & 1.4304 \\ 
		$\mathcal{P}^{\text{Alg}}_1$& 11 & 6982.5798 & 1.20 &173.93 & 1.7530 &1913.28 & 1.0000 \\
		$\mathcal{P}^{\text{Alg}}_2$& 9& 6975.5149 & 1.09 &353.81 &  3.5660 & 3184.27& 1.6643 \\ 
		$\mathcal{P}^{\text{Alg}}_3$&  5& 6911.0475 & 0.16 &2818.69 &  28.4085 & 14093.47 & 7.3661 \\ 
		$\mathcal{P}^{\text{Alg}}_4$&  4& 6923.4294 & 0.34 &1681.59 &  16.9481 & 6726.37 & 3.5156 \\ 
		\hline
	\end{tabular}
	\label{tab:results-random}
\end{table*}

\begin{figure}[t]
	\centering
	\subfloat[]{
		\includegraphics[width=.85\linewidth]{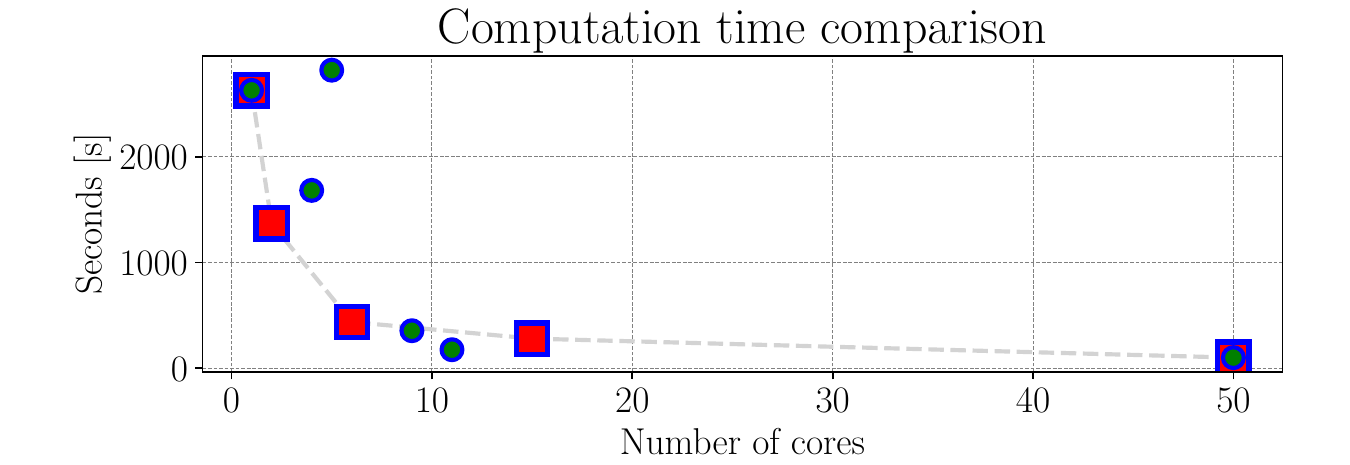}
		\label{fig:computationtimes-random}
	}
	\hfill
	\subfloat[]{
		\includegraphics[width=.85\linewidth]{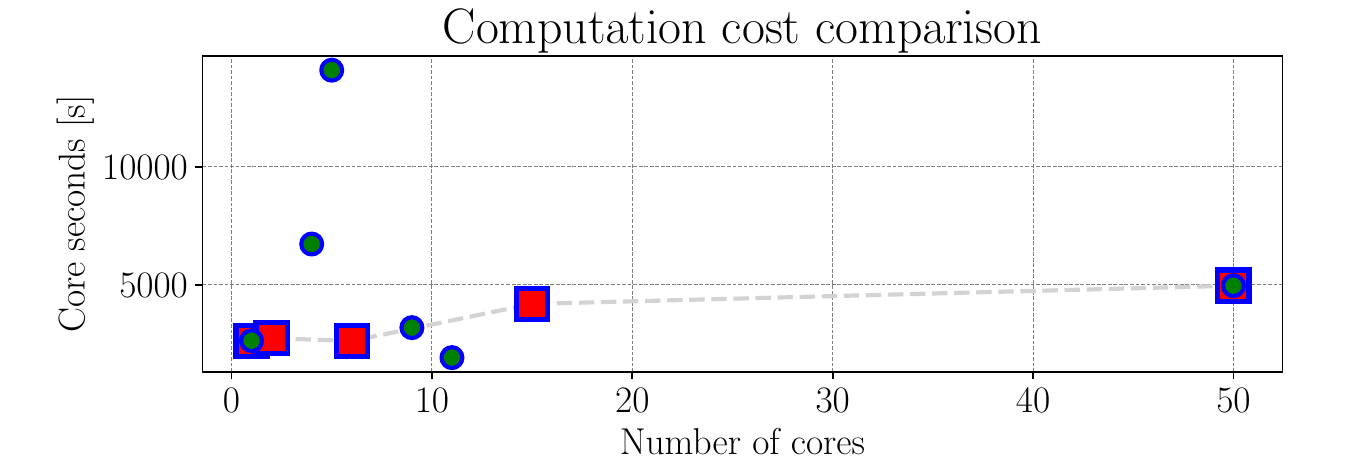}
		\label{fig:computationcost-random}
	}
	\caption{Computation time in [s] (Fig.\ \ref{fig:computationtimes-random}), and computation cost in core-seconds (Fig.\ \ref{fig:computationcost-random}) required to perform the simulation related to the number of required cores. The results for each strategy are reported by a marker, square-red for optimization-based partitioning, circle-green for algorithmic.}
\end{figure}

%


\section{Conclusions and Future Work}
This work establishes foundations for the systematic partitioning of networks of control systems. We have pursued this goal by formalizing the concept of equivalent graph of a dynamical system, extending analogous previous propositions; and introducing the novel concepts of fundamental system unit (FSU), composite system unit (CSU), aggregation operation, and control partition. This has allowed us to propose a generalized partitioning technique consisting of two steps: first we select FSUs; then we aggregate them to obtain a collection of CSUs, which is the control partition. The former step is achieved through a selection algorithm; while for the latter we propose both an optimization-based and an algorithmic approach. 
These new approaches have been validated through examples, and case studies on networks of linear and hybrid systems of different topologies. Empirical evidence shows that, through partitioning, we can achieve superior distributed control performance, computation efficiency, and lower costs. 

Future work could focus on improving the real-time applicability of these techniques and exploring the impact of the topology in distributed network control in terms of performance and stability.
Other areas for extensions are the inclusion in the equivalent graph of external signals or disturbances through additional nodes, and applying generalized partitioning to control strategies different from predictive control.




\bibliographystyle{IEEEtran}
\bibliography{IEEEabrv,bibliography.bib}
\section{Biography Section}
%
%
\begin{IEEEbiography}[{\includegraphics[width=1in,height=1.25in,clip,keepaspectratio]{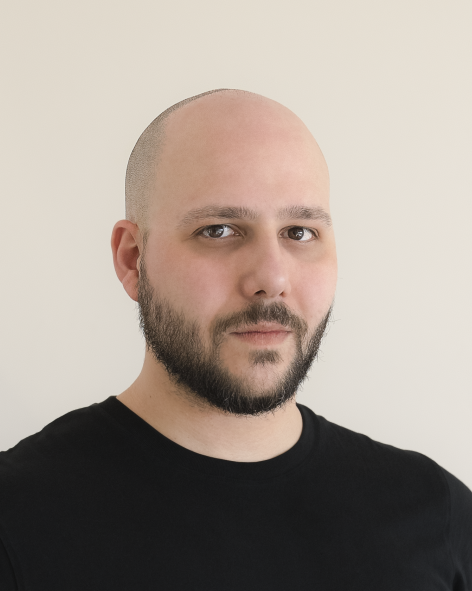}}]{Alessandro Riccardi} is a PhD candidate in the Delft Center for Systems and Control at the Delft University of Technology. He received a BS in automation engineering, and an MS in control engineering both from the University Sapienza of Rome. His research interests include predictive control and reinforcement learning for complex large-scale systems, and learning-based modeling of system dynamics.
\end{IEEEbiography}
\begin{IEEEbiography}[{\includegraphics[width=1in,height=1.25in,clip,keepaspectratio]{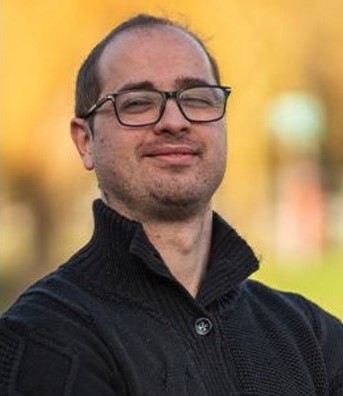}}]{Luca Laurenti}
	(Member, IEEE) is a tenure-track assistant professor at the Delft Center for Systems and Control at TU Delft and co-director of the  HERALD Delft AI Lab. He received his PhD from the Department of Computer Science, University of Oxford (UK). Luca has a background in stochastic systems, control theory, formal methods, and artificial intelligence. His research work focuses on developing data-driven systems provably robust to interactions with a dynamic and uncertain world.
\end{IEEEbiography}
\begin{IEEEbiography}[{\includegraphics[width=1in,height=1.25in,clip,keepaspectratio]{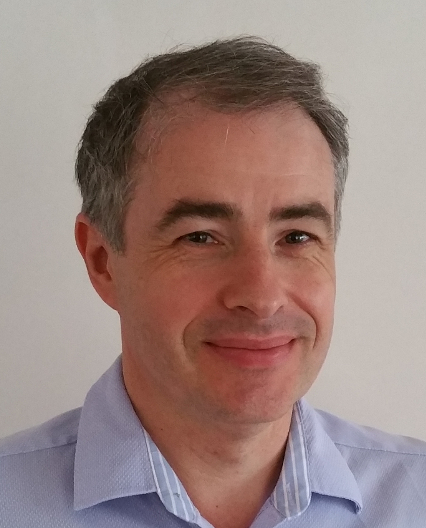}}]{Bart De Schutter}
	(IEEE member since 2008, senior member since 2010,
	fellow since 2019) is a full professor and head of
	department at the Delft Center for Systems and Control of Delft
	University of Technology in Delft, The Netherlands.
	
	Bart De Schutter is senior editor of the IEEE Transactions on
	Intelligent Transportation Systems.  His current research interests
	include integrated learning- and optimization-based control,
	with applications in transportation and energy networks.
\end{IEEEbiography}
\appendices
\section{Additional Profs}
Proof of \textbf{Prop.\ 1}.
\begin{proof}
The statement is verified by construction. Consider two CSUs, $\mathcal{S}_{1,k}=(\mathcal{V}_1,\mathcal{E}_{1,k},w_{1,k},\tilde{g}_1)$, $\mathcal{S}_{2,k}=(\mathcal{V}_2,\mathcal{E}_{2,k},w_{2,k},\tilde{g}_2)$, with $\mathcal{V}_1 = \mathcal{X}_1\cup\mathcal{U}_1$, $\mathcal{V}_2 = \mathcal{X}_2\cup\mathcal{U}_2$. By definition of CSU, for both $\mathcal{S}_{1,k}$ and $\mathcal{S}_{2,k}$, it holds that no edges are present from the set of nodes $\mathcal{U}_i$ to the set of states $\mathcal{X}_j$, for $i,j=1,2$, and $i\neq j$.
Now we define the CSU resulting from their aggregation as $\mathcal{S}_{(1,2),k}$ $=$ $\left.\left(\mathcal{S}_{1,k}\uplus\mathcal{S}_{2,k}\right)\right\lvert_{\mathcal{G}_k}$ $=$ $\left(\mathcal{V}_1\cup\mathcal{V}_2,\mathcal{E}_{1,k}\cup\mathcal{E}_{2,k}\cup\mathcal{E}_{(1,2),k},w_{(1,2),k},(\tilde{g}_1,\tilde{g}_2)\right)$, where $\mathcal{E}_{(1,2),k}$ contains the edges in $\mathcal{E}_k$ linking the two sets of nodes, but that are not present in the individual sets of edges. The operation does not add any edge connecting the set of input nodes $\mathcal{U}$ to the aggregated set of states $\mathcal{X}_1\cup\mathcal{X}_2$. Thus, the CSU resulting from the aggregation is a CSU according to Definition \ref{def:control-agent}.
\end{proof}
\section{Additional Figures}
\begin{figure}[h]
	\centering
	\includegraphics[width=0.9\linewidth]{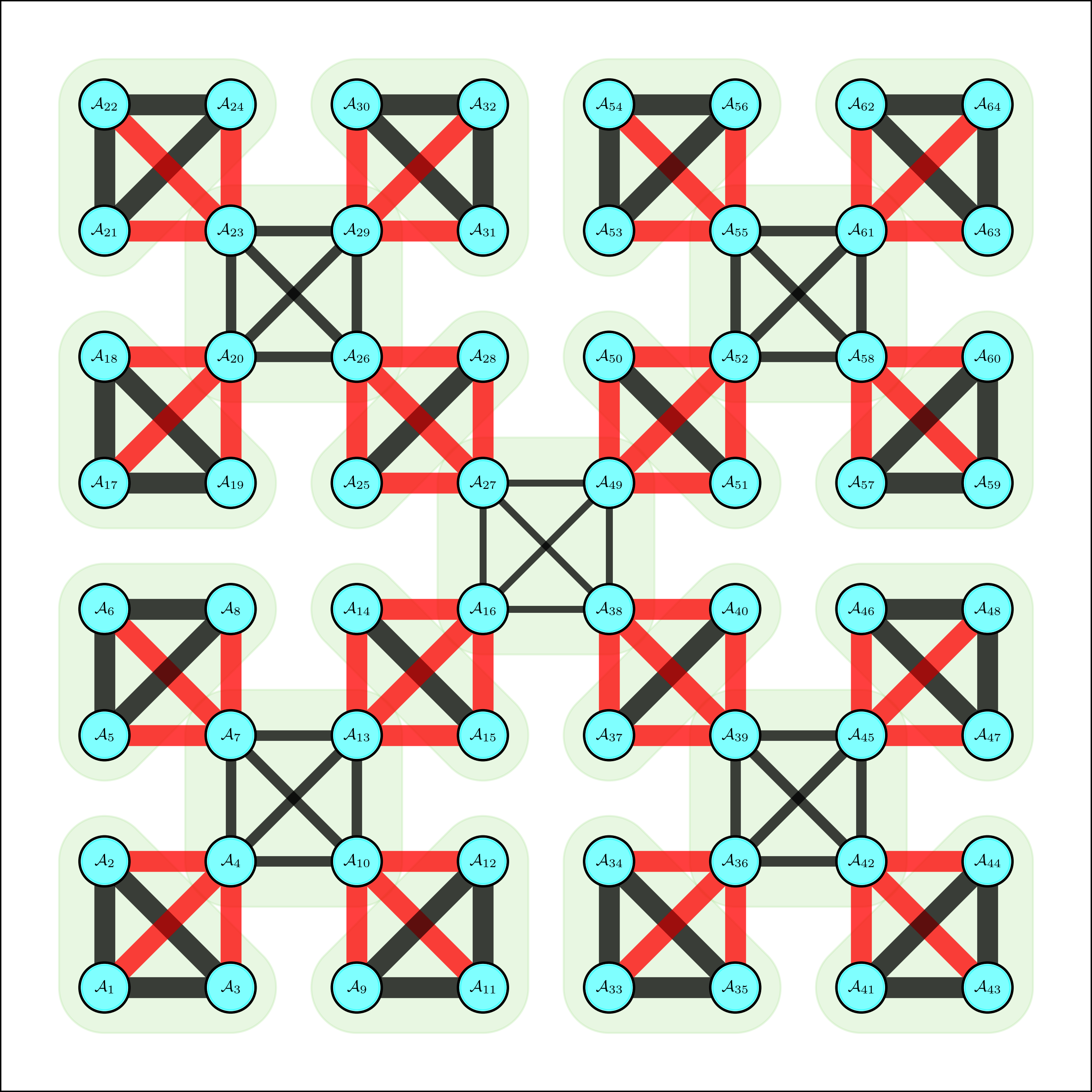}
	\caption[]{ A partition selected to maximize the interaction among CSUs used in Sec.\ \ref{subsec:ADMM-linear}}
	\label{fig:Modular_64_test1}
\end{figure}
%

\end{document}